\newif\ifSC
\SCtrue
\ifSC
\documentclass[onecolumn,draftclsnofoot,12pt]{IEEEtran}
\else
\documentclass[10pt,twocolumn]{IEEEtran}
\fi
\usepackage{siunitx}
\usepackage{cite}
\usepackage{caption}
\usepackage{subcaption}
\usepackage{amsmath,amssymb,amsfonts}
\usepackage[dvipsnames]{xcolor}
\usepackage{amssymb}

\usepackage{bbm}
\usepackage{bm}
\usepackage{amssymb,mathtools}
\usepackage{amsbsy}
\usepackage{multicol}
\usepackage{comment} 
\usepackage[labelformat=simple]{subcaption}

\usepackage[font=small]{subcaption}
\usepackage[font=small]{caption}
\usepackage{graphicx} 

\def\BibTeX{{\rm B\kern-.05em{\sc i\kern-.025em b}\kern-.08em
		T\kern-.1667em\lower.7ex\hbox{E}\kern-.125emX}}
\newcounter{relctr} 
\everydisplay\expandafter{\the\everydisplay\setcounter{relctr}{0}} 

\newcommand\numeq[1]%
{\stackrel{\scriptscriptstyle(\mkern-1.5mu#1\mkern-1.5mu)}{=}}
\AtBeginDocument{}

\newcommand{\x}{\mathbf{x}}
\newcommand{\y}{\mathbf{y}}
\newcommand{\z}{\mathbf{z}}

\newcommand{\road}{\bm{l}}

\newcommand{\1}{\mathbbm{1}}
\newcommand{\pt}{\mathrm{P}}

\newcommand{\drm}{\mathrm{d}}

\newcommand{\ob}{\mathrm{o}}

	\definecolor{ao(english)}{rgb}{0.0, 0.5, 0.0}

\newcommand{\dv}{\mathrm{d}}
\newcommand{\bt}{\mathbf{b}}

\newcommand{\mrm}{\mathrm{m}}

\newcommand{\prm}{\mathrm{p}}

\newcommand{\typical}{\mathrm{t}}

\newcommand{\A}{\mathcal{A}}

\newcommand{\R}{\mathbb{R}}

\newcommand{\matern}{Mat\'ern~}
\newcommand{\rfirst}{a}

\newcommand{\adjustfigspace}{\vspace{-.2in}}
\newcommand{\Vor}{\mathsf{V}}
\newcommand{\bs}{\mathrm{b}}
\newcommand{\roads}{\mathrm{L}}
\newcommand{\BSPP}{\Phi_{\bs}}

\newcommand{\eg}{{\em e.g.}~}
\renewcommand{\bt}{\mathsf{B}}
\newcommand{\oboneD}{\bm{o}}
\renewcommand{\road}{{\ell}}
\newcommand{\roadf}{\mathsf{L}}
\newcommand{\bellf}{\mathsf{b}}

\newcommand{\densityPTS}{\lambda_\mrm}
\newcommand{\densityITS}{\kappa}
\renewcommand{\densityITS}{\lambda_\mathrm{i}}
\newcommand{\densityRoads}{\lambda_{\roads}}
\newcommand{\densityBS}{\lambda_{\bs}}

\newcommand{\vehPP}{\Psi_{\mrm}}
\newcommand{\roadPP}{\Phi}
\newcommand{\platoonPP}{\psi_{\mrm,{i}}}

\newcommand{\platooncentersPP}{\Psi_{i}}
\newcommand{\dauP}{\Delta}
\renewcommand{\dauP}{\Omega}
\newcommand{\dauPuntrans}{\overline{\Delta}}
\renewcommand{\dauPuntrans}{\overline{\Omega}}
\newcommand{\platoontransPP}{\Psi_{\road_i}}
\newcommand{\pgfl}[2]{\mathcal{G}_{#1}\left(#2\right)}

\renewcommand{\roadPP}{\Phi_\roads}
\renewcommand{\platoonPP}{\Psi_i}
\newcommand{\platoonPPuntrans}{\overline{\Psi}_i}

\renewcommand{\platooncentersPP}{\overline{\platoonPP}}

\renewcommand{\platooncentersPP}{\dot{\platoonPP}}
\newcommand{\platooncentersPPwo}{\dot{\Psi}}
\renewcommand{\platooncentersPP}{\breve{{\platoonPPuntrans}}}
\renewcommand{\platooncentersPPwo}{\breve{\overline{\Psi}}}
\renewcommand{\platooncentersPP}{{{\platoonPPuntrans}^{\mathrm{(p)}}}}
\renewcommand{\platooncentersPP}{\overline{\Psi}_i^{_\mathrm{(p)}}}
\renewcommand{\platooncentersPPwo}{\overline{\Psi}^{_\mathrm{(p)}}}

\newcommand{\untrans}[1]{\widetilde{\bm{#1}}}

\renewcommand{\untrans}[1]{\overline{\bm{#1}}}

\renewcommand{\hat}[1]{\widehat{#1}}

\newcommand{\off}{\mathrm{p}_{\mathrm{off}}}
\renewcommand{\off}{p_{\mathrm{off}}}

\newcommand{\avg}{\mathrm{p}_{<}}
\renewcommand{\avg}{{p}_{\mathrm{u}}}
\newcommand{\kavg}{s_\mathrm{avg}}

\newcommand{\Pmin}{\mathrm{P}_{1}}
\newcommand{\Pavg}{\mathrm{P}_{<}}

\newcommand{\mavg}{m_\mathrm{avg}}
\renewcommand{\Pmin}{p_{1}}
\renewcommand{\Pavg}{p_{\mathrm{s}}}

\newcommand{\rsecond}{b}

\usepackage{bbm}
\usepackage{amssymb}
\usepackage{amsmath,amsthm}
\usepackage{color}
\usepackage{amsfonts}
\usepackage{graphicx}
\usepackage{verbatim}
\usepackage{epstopdf}
\usepackage{cite}
\usepackage{tabulary}
\usepackage{mathtools}
\usepackage{enumitem}

\newcommand{\set}[1]{\mathsf{#1}}

\newcommand{\palmexpectx}[2]{\mathbb{E}^{#2}{\left[#1\right]}}

\newcommand{\fracS}[2]{#1/#2}
\newcommand\expect[1]{\mathbb{E}\left[#1\right]}
\newcommand\prob[1]{\mathbb{P}\left[#1\right]}

\newcommand\indside[1]{\mathbbm{1}\left({#1}\right)}

\newcommand{\SIR}{\text{SIR}}

\newcommand{\laplace}[1]{\mathcal{L}_{#1} }
\newcommand{\ie}{{\em i.e.}~}

\newcommand{\Rc}{\mathrm{r_c}}

\newcommand{\pu}{\mathrm{P}}

\newcommand\lambdaPT{\lambda_\pu}

\newcommand{\expS}[1]{\exp{\left(#1\right)}}

\def\home{\hbox{\kern3pt \vbox to13pt{}%
   \pdfliteral{q 0 0 m 0 5 l 5 10 l 10 5 l 10 0 l 7 0 l 7 5 l 3 5 l 3 0 l f
               1 j 1 J -2 5 m 5 12 l 12 5 l S Q }%
   \kern 13pt}}

\renewcommand{\SIR}{\mathtt{SIR}}

\newtheorem{theorem}{Theorem}
\newtheorem{lemma}{Lemma}
\newtheorem{definition}{Definition}
\newtheorem{corollary}{Corollary}[theorem]

\begin{document}
	
	\title{Fundamentals of Vehicular Communication Networks with Vehicle Platoons}
	\author{Kaushlendra Pandey,  Kanaka Raju Perumalla, Abhishek K. Gupta, Harpreet S. Dhillon \,\,\vspace{-2.1em}
		\thanks{K. Pandey, K. Raju Perumalla and A. K. Gupta are with IIT Kanpur, India, 208016. Email:\{kpandey,pkraju,gkrabhi\}@iitk.ac.in. H. S. Dhillon is with Wireless@VT, Bradley Department of Electrical and Computer Engineering, Virginia Tech, Blacksburg, VA 24061 (Email: hdhillon@vt.edu).} 
	}
	\maketitle
	\begin{abstract}
Vehicular platooning is a promising way to facilitate efficient movement of vehicles with a shared route. Despite its relevance, the interplay of platooning and the communication performance in the
resulting vehicular network (VN) is largely unexplored. Inspired by this, we develop a comprehensive approach to statistical modeling and system-level analysis of VNs with platooned traffic. Modeling the network of roads using the by-now well-accepted Poisson line process (PLP), we place vehicles on each road according to an independent \matern cluster process (MCP) that jointly captures randomness in the locations of platoons on the roads and vehicles within each platoon. The resulting {\em triply-stochastic} point process is a PLP-driven-Cox process, which we term the PLP-MCP. We first present this new point process’s distribution and derive several fundamental properties essential for the resulting VN’s analysis. Assuming that the cellular base-stations (BSs) are distributed as a Poisson point process (PPP), we derive the distribution of the loads served by the typical BS and the BS associated with the typical user. In deriving the latter, we also present a new approach to deriving the length distribution of a tagged chord in a Poisson Voronoi tessellation. Using the derived results, we present the rate coverage of the typical user while considering partial loading of the BSs. We also provide a comparative analysis of VNs with and without platooning of traffic.
	\end{abstract}
	\IEEEpeerreviewmaketitle

\section{Introduction}
Vehicular platooning refers to the cooperative movement of closely located vehicles having a shared route or a part of route. As part of intelligent transportation systems, platooning has enormous potential for collision avoidance among vehicles,  optimization of the road capacity and fuel consumption, and reduction in pollutant gases including CO$_{2}$ emissions\cite{zeadally2020vehicular,jia2015survey}. 
Platooning and vehicular communication have a two way relationship. 
{On one hand, platooning almost always ensures line-of-sight between two proximate vehicles, thereby improving the reliability of vehicle-to-vehicle (V2V) communication between them compared to independently moving vehicles} \cite{perfecto2017millimeter}. Such V2V communications can help in collision and hazard warning and traffic planning  \cite{ahmed2018cooperative}. Further, if one vehicle in the platoon is able to receive information via vehicle-to-infrastructure (V2I) communication, V2V communication can help relay this data to all vehicles in the platoon. 
On the other hand, vehicular communication is also essential in enabling platooning to reduce collision risks due to smaller intra-vehicular distance. 
{Given the intertwined nature of these two seemingly disparate ideas, it is essential to understand their synergism, which we do here by carefully integrating platooning in the system-level analysis of vehicular networks.
} 
\subsection{Related work}
Recently, there has been a significant interest in studying different types of vehicular communication including V2V and V2I. { Interested readers are advised to refer to \cite{ahmed2018cooperative,siegel2017survey}, and the references therein, for a comprehensive survey of this research direction. In this paper, our specific interest is on the system-level analysis of vehicular communications networks using stochastic geometry, which has attracted significant attention recently, e.g, see \cite{blaszczyszyn2009maximizing,blaszczyszyn2013stochastic,wang2018mmwave,wang2017blockage,
		Baccelli96stochasticgeometry}. However, the focus of almost all of this prior work has been on conventional {\em non-platooning traffic scenarios} (N-PTS), where vehicles do not form platoons and hence move without any coordination with each other.} 
For instance, in \cite{blaszczyszyn2009maximizing,blaszczyszyn2013stochastic} 
authors modeled the vehicular traffic on a fixed road by a $1$D Poisson point process (PPP). To incorporate multiple road vehicular traffic, 
\cite{wang2018mmwave,wang2017blockage} considered grid type urban roads (roads are either perpendicular, or parallel to the $\mathsf{x}$-axis) modeled using the  Manhattan Poisson line process (MPLP). Each road has an independent vehicular traffic distributed as $1$D PPP. The authors analyzed the blockage and coverage in such networks. To include the {irregularity in the layout of roads,} in \cite{Baccelli96stochasticgeometry}, authors suggested to model roads as Poisson line process (PLP) and vehicles as 1D PPP on each road. In this model, the combined vehicular traffic across roads forms a Cox process that can be termed PLP-PPP (\ie a PLP driven PPP). A thorough investigation of various  properties of PLP-PPP and its applications to vehicular communications was presented in \cite{dhillon2020poisson}.

 A vehicular communication network consists of vehicular traffic overlaid with a cellular network to provide infrastructure connectivity to vehicular traffic. Such a network with {N-PTS} can be modeled using PLP-PPP overlaid with an independent PPP modeling the locations of BSs, owing to the mathematical tractability of these processes. 
In \cite{guha2016cellular,chetlur2018coverage,chetlur2018characterization}, authors derived the distribution of signal-to-noise-plus-interference ratio (SINR) for similar models. In \cite{chetlur2020load,sial2019stochastic,choi2018analytical} authors derived the SINR distribution for vehicle-to-everything (V2X) networks consisting of communications between different types of network entities, such as between BSs and vehicles, and roadside units and vehicles. 
Another important  metric dictating the overall performance of a network is the rate distribution of the typical user. The achievable rate  depends critically on the per-BS load, \ie the number of vehicles present in the BS's serving region.
In \cite{chetlur2020load}, authors derived the distribution of the per BS load and per-user rate for {N-PTS}. In \cite{chetlur2018success}, the area spectral efficiency for the N-PTS was presented. In \cite{chetlur2019coverage}, authors derived the rate coverage for cellular vehicle-to-everything (V2X) networks for N-PTS.

Although past works have analyzed the vehicular communication network with N-PTS, analytical tools have not been fully explored yet to study the {\em platooned vehicular traffic scenario} (PTS) and its impact on the performance of a vehicular communication. Consequently, there is limited work focusing on the analysis of PTS \cite{shao2015performance,yi2019modeling,zeng2019joint}. For example, in \cite{shao2015performance}, authors considered a single road vehicular traffic consisting of independent vehicles and platoons, both modeled as points of $1$D PPP and derived the probability that vehicles can communicate with each other. In \cite{yi2019modeling}, authors considered platooned traffic on a single road with BSs deployed on the side of the road and derived approximate coverage. In \cite{zeng2019joint}, authors considered platooned traffic on a single road with road side BSs and performed a joint communication and control analysis to study the stability and delay in the network.
One main limitation of the above works is that they considered vehicular traffic on a single road. In practice, the ``support'' of a vehicular network is a complicated layout of roads that needs to be accounted for and is one of the key reasons for the popularity of the PLP-based models. A vehicular traffic on such a road network is further complicated by the randomness in the number and locations of platoons and locations of vehicles in each platoon. As indicated above, the wireless performance of a vehicular user depends critically on the per BS load which has not been analyzed in the previous work for PTS. Overall, the interplay of platooning and the vehicular network performance is largely unexplored from the perspective of rigorous system-level analysis. 
This paper attempts to bridge this gap. In particular, we try to explore how we can model a complete vehicular communication network consisting of a 2D network of BSs and the vehicles moving in platoons and analyze the performance of this network in terms of load and rate distribution.
\subsection{Contributions}
In this paper, we develop an analytical framework for a vehicular communication system (in particular, a V2I scenario) with platooned traffic. We propose a novel point process for modeling  the platoon movement of vehicles. We then examine the impact of platooning on V2I communication by observing the load that appears on the infrastructure network. We present a  comparative study of the rate coverage for PTS and N-PTS. The important contributions of this paper are listed below.

\begin{enumerate}
	\item We propose a novel point process termed PLP-MCP for the modeling and analysis of the platooned movement of the vehicles. It is a Cox process driven by the PLP that captures three layers of randomness: (i) irregularity in the road layout, (ii) randomness in the locations of the platoons, and (iii) randomness in the locations of vehicles within a platoon. In this sense, this process can be thought of as a {\em triply-stochastic} process that generalizes {\em doubly-stochastic} PLP-PPP used in the literature \cite{dhillon2020poisson}. We present its distribution and key properties  essential for the analysis of the vehicular traffic. 
	\item We then present an analytical framework to characterize the performance of the typical user in a V2I communication network consisting of BSs and platooned traffic. 
	\item We derive the load distribution for the typical and the tagged BSs along with the means and variances. Here, tagged BS is the BS that serves the typical user. As a key intermediate result, we derive a new expression for the distribution of the tagged chord in the Voronoi cell of the tagged BS.
	\item Using the derived results, we present the rate coverage of the typical user while considering the partial loading of BSs. We perform a comparative analysis of load distribution and rate coverage of communication systems with platooned movement with non-platooned movement to understand the impact of vehicular platooning.  
\end{enumerate}

\textit{Notation:}
			Vectors in  $\R$ are denoted by  bold \textit{italic} style letters  (\eg $\bm{x}$) with their norms as $|\bm{x}|$. Similarly, vectors in $\R^{2}$ are denoted by bold style letters (\eg $\x$) with their norms as {$\|\x\|$}. The origin is $\ob\equiv(0,0)$.
			Let  $\bt_{1}(\bm{x},r)$ and $\bt_{2}(\x,r)$ denote a 1D and 2D ball centered at $\bm{x}$ and $\x$ of radius $r$. 
			Let $\road=\roadf(\rho,\phi)$ denote a line in $\R^{2}$ in Hesse normal form{,} \ie the normal segment from origin to the line is of length $\rho$ and makes angle $\phi$ with respect to the x-axis. The point $\left(\rho \cos \phi,\rho \sin \phi\right)$  is the nearest point on the line $\roadf(\rho,\phi)$ from the origin termed the base. The line $\roadf(\rho,\phi)$ can also be represented as an element $\left(\rho,\phi\right)$ of the set $\mathbf{C^{*}}\equiv \R\times[0,\,\pi{)}$. We term the element $\left(\rho,\phi\right)$ as L-atom and $\mathbf{C^{*}}$ as L-space. Further, $f_\ell()$ denotes the transformation of  $\roadf(0,0)$ to the line $\ell=\roadf(\rho_\ell,\phi_\ell)$ given as
			\begin{align}
			f_{\ell}(\untrans{x})=\left(\rho_{\ell} \cos \phi_{\ell} + \untrans{x} \sin \phi_{\ell},\rho_{\road}\sin \phi_{\road} -\untrans{x} \cos \phi_{\ell}\right).
			\end{align}		
This means that if  $\untrans{x}$ is a scalar quantity denoting the  location of a point in the line $\ell$ relative to its base, its $2$D coordinates (\ie absolute location in $\R^2$) are given as ${\x}=f_{\ell}(\untrans{x})$. For a set $\set{A}$, $|\set{A}|$ denotes its Lebesgue measure  in its respective dimension, for example $|\bt_{1}\left(\ob,r\right)|=2r$. The volume $\A_{1}(\rfirst,\rsecond,x)$ of the intersection of two $1$D balls $\bt_{1}(\oboneD,\rfirst)$ and $\bt_{1}(\bm{x}, \rsecond)$ is given as \cite{pandey2019contact}
		\begin{align}
		\A_{1}(\rfirst,\rsecond,x)=
	\begin{cases}
		2\min(\rfirst,\rsecond),\,\, & \text{if }0\leq x \leq |\rfirst-\rsecond|\\
		\rfirst+\rsecond- x,\,\,&
		\text{if } |\rfirst\!-\! \rsecond|\!<\!x\!\le\!\rfirst\!+\!\rsecond.
	\end{cases}
	\end{align}
		The PDF of the generalized Gamma distribution with parameters $a_{1},b_{1},c_{1}$ is denoted  by 
		\begin{align}\label{gammadistribution}
			\tilde{g}_{X}(x;a_{1},b_{1},c_{1})=a_{1}b_{1}^{c_1/a_{1}}\left(\Gamma\left(c_{1}/a_{1}\right)\right)^{-1}x^{c_{1}-1}e^{-b_{1}x^{a_{1}}}.
		\end{align}	
		 Further, $\bellf{}(\cdot)$ denotes the Bell's polynomial \cite{frucht1965polynomios}. For a point process (PP) $\Psi$, the notation $\Psi(\mathsf{C})$ denotes the number of points of $\Psi$ falling inside set $\mathsf{C}$. The PGF of any integer-valued random variable (RV) $X$ is denoted by $\mathcal{P}_{X}(\cdot)$. The expected value of RV  $X$ is denoted by $\mathbb{E}[X]$. Further, $\beta(r)=2\min(r,a)$. The notation $\widetilde{(\cdot)}$ denotes the approximated variable and $\hat{(\cdot)}$ denotes reduced Palm version.  If $A$ and $B$ are two RVs, the notation $A\stackrel{(d)}=B$ means that $A$ and $B$ have the same distribution.		 
		 
		\section{Modeling of Platooned Vehicles using PLP-MCP}\label{sec:vehicletraffic}
		In this paper, we introduce a novel point process PLP-MCP to model platooned vehicles on a network on roads. The system model is as follows. 
		\vspace{-.2in}
\subsection{Road network}  
The network of roads can be modeled by a PLP $\roadPP=\{\road_{1},\road_{2},\cdots\}$ with density $\densityRoads$ where  $\road_i$ denotes the $i$th road \cite{dhillon2020poisson}. The $i$th line $\road_i\in{\roadPP}$ can be denoted by the L-atom $a_i=(\rho_{\road_{i}},\phi_{\road_{i}})$ in the L-space $\mathbf{C^{*}}$. The L-atoms $a_i$'s form a PPP in $\mathbf{C^{*}}$ with density $\densityRoads$. This means that the mean number of lines hitting a convex body $\mathsf{K}$ with perimeter $L(\mathsf{K})$ is $\densityRoads L(\mathsf{K})$ \cite{dhillon2020poisson}.  
\vspace{-.1in}
\subsection{Platooned vehicles}
For each road $\road_i$, vehicular platoons can be seen as the clusters of vehicles in a finite spread. Since the vehicles are usually uniformly distributed in the respective platoons, it is natural to model the resulting traffic on each road using MCPs. We model the vehicles on the  road $\road_i$ by an independent MCP $\platoonPP$ with parent PP density $\lambdaPT$, mean number of points per cluster $m$ and cluster radius $a$. In particular,  the platoon centers are distributed as the parent PP $\platooncentersPP$. For a platoon centered at $\bm{x}_{j,i} \in \platooncentersPP$, the constituent vehicles are distributed as the PPP $\dauPuntrans_{\bm{x}_{j,i}}$ in $a-$neighborhood of it. Let $\mu_\mrm$ denote the per-road vehicular density \ie $\mu_\mrm=m\lambda_\pt$.

 The  locations of all vehicles form a new PP, which we term as PLP-MCP. It can be formally defined as follows.  
	\begin{definition}	[PLP-MCP]
		Let $\roadPP=\{\road_{1},\road_{2},\cdots\}$ be a {PLP}  with density $\densityRoads$ with the $i$th line $\road_i=\roadf(\rho_{\road_{i}},\phi_{\road_{i}})$. 
		Let $\{\platoonPPuntrans\}$ be a set of independent and identically distributed $1$D MCP in $\R$ with parameter $\left(m,\lambdaPT,a\right)$ such that \vspace{-.1in}
\begin{align*}
	&\platoonPPuntrans=\bigcup\nolimits_{\bm{x}_{j,i}\in \platooncentersPP}\dauPuntrans_{\bm{x}_{j,i}},
\end{align*} 
where $\platooncentersPP$ is a PPP with density $\lambdaPT$.  $\platooncentersPP$ is called the parent point process of $\platoonPPuntrans$ as it  consists of parent points $\bm{x}_{j,i}\in \R$. Further, $\dauPuntrans_{\bm{x}_{j,i}}$ denotes the daughter PP of $\bm{x}_{j,i}$ and is a PPP with density $\lambda_\mathrm{d}=m/(2a)$ in $\bt_1(\bm{x}_{j,i},a)$.
We assign $i$th MCP $\platoonPPuntrans$ to the $ i$th line $\road_i$ and transform the points of $\platoonPPuntrans$ to be on the line to get 
\begin{align}
	&\platoontransPP=\bigcup\nolimits_{\bm{x}_{j,i}\in\platooncentersPP }
	\{ {\z}_{k,j,i}=
	f_{\road_i}
	(\untrans{z}_{k,j,i}):\untrans{z}_{k,j,i}\in \dauPuntrans_{\bm{x}_{j,i}} \}=
	\bigcup\nolimits_{\bm{x}_{j,i}\in\platooncentersPP }
	 \dauP_{\bm{x}_{j,i}},
\end{align}
where $ \dauP_{\bm{x}_{j,i}}$ represents   $\dauPuntrans_{\bm{x}_{j,i}}$ transformed on line $\ell_i$.

Now, a PLP-MCP $\vehPP$ is defined as the union of all $\platoontransPP$'s 
	\ie 
\begin{align}
\vehPP=\bigcup\nolimits_{\road_i\in\roadPP} \platoontransPP,
\end{align}
and includes all the points located on every line of $\roadPP$. 
\end{definition}		
	Hence, the   platoon vehicular traffic can be modeled using points of the proposed PLP-MCP $\vehPP$. The absolute location of $k$th vehicles in $j$th platoon of $i$th road is given as ${\z}_{k,j,i}$. The {\em typical point} of the PLP-MCP $\vehPP$ denotes \textit{the typical vehicle} \cite{SGBook2022}.  A line $\road_{\ob}$ of $\roadPP$ that passes through the typical point 
 of PP ($\vehPP$) is termed as {\em the tagged line}. Here, $\road_\ob=\roadf(0,\phi)$ with $\rho=0$ and $\phi$ is a uniform RV between $0$ to $\pi$.

\vspace{-.1in}
\subsection{Properties of PLP-MCP}

We now describe some key properties of the PLP-MCP that are helpful in the analysis of vehicular communication.
\begin{enumerate}
\item {\textit{Stationarity:}}
			The PLP-MCP $\Psi_{\mrm}$ is a stationary PP. The stationarity of $\Psi_{\mrm}$ follows from the stationarity of PLP and $1$D MCP.
\item {\textit{Density:}}
		The density $\densityPTS$ of $\Psi_{\mrm}$ is $m\lambdaPT{\densityRoads}\,\pi$. The density of $\Psi_{\mrm}$ can be derived by counting the mean number of points in a unit area using the Campbell’s theorem \cite{SGBook2022}.
\item
It is  a Cox process driven by a PLP.
\end{enumerate}				
				For a stationary PP, we can take the typical point  at the origin \cite{SGBook2022}.
Further, if the typical point is located at the origin, the tagged line $\road_{\ob}$ passes through the origin.

\vspace{-.1in}
\section{Characterization of PLP-MCP}
In this section, we will present several key properties of  the proposed PLP-MCP. 
\vspace{-.1in}
\subsection{Extended Slivnyak Theorem}
Since PLP-MCP is derived from PLP (which is a PPP in L-space),  Slivnyak theorem can be extended to describe the conditional distribution of PLP-MCP. Even though this extension is not overly challenging, we decided to present it separately upfront so that we can easily refer to it throughout the paper rather than repeating this same argument everywhere.
\begin{lemma}(Extended Slivnyak Theorem.)\label{thm:slivE}
Conditioned on the typical point $\z_\ob$, the distribution of the rest of the PLP-MCP $\vehPP$ is equal to the distribution of an independent copy of $\vehPP$ superposed with an independent copy of the MCP $\Psi_{\road_\ob}$ on the tagged line $\road_\ob$ and an independent copy of the cluster PPP $\dauP_{\bm{x}_\ob}$ (which is $\dauPuntrans_{\bm{x}_\ob}$ transformed on $\road_\ob$). Here, $\bm{x}_\ob$ denotes the parent point of the typical point and distributed uniformly in the 1D $a-$neighborhood of $\bm{z}_\ob=f_{\road_\ob}^{-1}(\z_\ob)$. In other words, \vspace{-.2in}
\begin{align}
\vehPP^{!}|(\z_\ob\in \vehPP) \stackrel{(d)}{=} \vehPP \cup \Psi_{\ell_\ob} \cup \dauP_{\bm{x}_\ob}.\label{eq:slivE}
\end{align}
\end{lemma}
\begin{proof}
Conditioning on the occurrence of the typical point fixes the tagged parent point $\bm{x}_\ob$ and the tagged line $\road_\ob$. Since PLP $\roadPP$ is a PPP in L-space, conditioned on $\road_\ob$, $\roadPP$ is equivalent to the union of $\road_\ob$ 
and an independent copy of $\roadPP$ (which generates an independent copy of $\vehPP$). Now,  $\platooncentersPPwo_{\road_\ob}$ is also a PPP containing $\bm{x}_\ob$, hence, conditioned on $\bm{x}_\ob$, it is equivalent to the union of $\bm{x}_\ob$ and an independent copy of $\platooncentersPPwo_{\road_\ob}$ which generates $\Psi_{\road_\ob}$. Again, $\dauP_{\bm{x}_\ob}$ is a PPP, hence, conditioned on $\z_\ob$, $\dauP_{\bm{x}_\ob}$ is equivalent to the union of $\bm{x_0}$ and an independent  copy of $\dauP_{\bm{x}_\ob}$. 
\end{proof}
\vspace{-.2in}
\subsection{Laplace functional}
Since the distribution of a PP is completely characterized by its Laplace functional (LF), we now derive the LF for PLP-MCP. We will first require the PGFL of the MCP $\platoontransPP$ transformed on the line $\road_i$ which is given in Lemma \ref{lem:PGFLMCP}. Let there be a function $v:\,\R^2\rightarrow \left[0,\,1\right]$.
\begin{lemma}\label{lem:PGFLMCP}
The PGFL of  $\Psi_\road$ on road $\road$ is given as
		\begin{align}\label{PGFLMCP}
			\pgfl{\Psi_{\road},\road}{v}=\exp\left(-\lambda_{\pt} \int_{\R} \left(1-\mathcal{H}_{\bm{x},\road}(v)\right)\dv \bm{x}\right),
		\end{align}
	where
$\displaystyle \mathcal{H}_{\bm{x},\road}(v)=\exp\left(-\lambda_{\drm} \int_{\bt_{1}(\ob,a)} \left(1-\left(v\circ f_{\road}\right)(\bm{x}+\bm{y})\right) \dv \bm{y}\right).$ 
 Under reduced Palm (\ie conditioned on occurrence of a point  at ${\z}_{\ob}$ excluding ${\z}_{\ob}$), PGFL of  $\Psi_\road$ on road $\road$  is		
		\begin{align}\label{under_palm}
			\mathcal{G}_{\Psi_{\road},\road}^{!{\z}_{\ob}}(v)	&=\pgfl{\Psi_{\road},\road}{v}{(2a)^{-1}}\int_{\bt_1(f_\ell^{-1}(\z_\ob),a)}\mathcal{H}_{-\bm{x}_{\ob},\road}(v)\dv \bm{x}_{\ob}.
		\end{align}
		Here, $\bm{x}_{\ob}$ denotes the untransformed center of the parent cluster of ${\z}_{\ob}$. 
\end{lemma}
			
\vspace{-.1in}		
We now derive the LF for PLP-MCP	which is given in the following two Theorems.  See Appendix \ref{Laplcefunctional} for the {proofs}.

\vspace{-.1in}			
		\begin{theorem}
			The LF for  $\Psi_{\mrm}$ is given as 
			\begin{align}
			\laplace{\vehPP}(v)=\expect{e^{-\sum_{{\z}\in\vehPP} v({\z})}}=
		{	\exp\left(-{\densityRoads}
			\int_{\R}\int_{0}^{\pi}\left(1-{G_{\Psi_{\roadf (\rho,\phi)},\roadf(\rho,\phi)}(e^{-v})}\right)\dv \rho\,\dv \phi \right)},\end{align}
				where  $G_{\Psi_{\road},\road}(v)$ is given in \eqref{PGFLMCP}.
		\end{theorem}
					
\vspace{-.2in}
		\begin{theorem}
				The LF for  $\Psi_{\mrm}$ under the  reduced Palm distribution is
				\begin{align}
				\laplace{\vehPP}^{!{\ob}}(v)=
				\palmexpectx{e^{-\sum_{{\z}\in\vehPP} v({\z})}}{!\ob}={\laplace{\Psi_{\mrm}}(v)
					\int_{0}^{\pi}\pi^{-1}G_{\Psi_{\roadf(0,\phi)},\roadf(0,\phi)}^{!{\ob}}(e^{-v})\,\dv \phi.}\end{align}
		\end{theorem}	
\vspace{-.2in}
\subsection{Distribution of number of points (vehicles) of $\vehPP$ in a set}
The PP $\vehPP$ can also be characterized by the distribution of the number  of its points  in a set which is crucial in computing the load distribution in vehicular communication network which will be discussed in the next section. 
To derive this distribution, we will first require the PGF of the number $N_\road$ of points of the MCP $\Psi_\road$ on the line $\road$ which is given in Lemma \ref{lem:PGF_N_MCP}. 
\begin{lemma}
\label{lem:PGF_N_MCP}
		Let $\Psi_{\road}$ denotes a $1$D MCP on line $\road=\roadf\left(\rho,\phi\right)$. 
		The PGF for the number $N_{\road}$ of points  of $\Psi_{\road}$  falling inside $\bt_{2}(\ob,r)$ is
		\begin{align}\label{PGFMCP}
			\mathcal{P}_{N_{\road}}(s,r)&={\exp\left(g(s,\sqrt{r^{2}-\rho^{2}})\right)}, \\
\label{g}
		\text{\hspace{-1.05in}where }	\hspace{.35in}g(s,t)&	=2\lambdaPT\left[\left|{t}-a\right|e^{\lambda_{\drm}\beta\left(t\right)(s-1)}-\left(t+a\right)+\fracS{(e^{\lambda_{\drm}(s-1)\beta\left(t\right)}-1)}{(\lambda_{\drm}(s-1))}\right].
		\end{align}
	
\end{lemma}
Note that $\rho=0$ gives the PGF of $N_{\road}$ when the line passes through the origin with an angle of $\phi$.	
{The} $k$th derivative of $g(s,t)$ with respect to $s$ is given as 
	\begin{align}\label{g(0,rho)}
	&g^{{(k)}}(s,t)=2\lambda_{\pt}\left[(\lambda_{\drm}\beta(t))^k|t-a|e^{(s-1)\lambda_{\drm}\beta(t)}\right.\nonumber\\
	&\left. 
	+\frac{1}{\lambda_{\drm}}\left(\sum_{j=0}^{k}{k\choose j}\frac{j!(-1)^{j}}{(s-1)^{j+1}}(\lambda_{\drm}\beta(t))^{k-j}e^{(s-1)\lambda_{\drm}\beta(t)}-\frac{k!(-1)^{k}}{(s-1)^{k+1}}\right)\right].
	\end{align}
{To derive the mean and variance, we need $\lim_{s\rightarrow1}g^{(k)}(s,t)$. Let $\lim_{s\rightarrow1}g^{(k)}(s,t)=\kappa(t,k)$ which is given as
	\begin{align}\label{kappaequation}
		\lim_{s\rightarrow1}g^{(k)}(s,t)=\kappa(t,k)=2\lambdaPT\left(\lambda_{\drm}\beta(t)\right)^{k}\left[\vert t-a \vert+\beta(t)/(k+1)\right].
\end{align}}
We now present the distribution of the number $S(r)$  of points of $\vehPP$ in a 2D ball of radius $r$ \ie  $S(r)=\vehPP(\bt_{2}(\ob,r))$ in terms of its PGF and PMF along with its mean and variance. 
Note that the PMF, the mean and the variance of a discrete RV $X$ can be computed from its PGF using the following relation
\begin{align}
p_X(k)=\mathbb{P}[X=k]&=\frac{1}{k!}\left[{\mathcal{P}^{(k)}_{X}(s,r)}\right]_{s=0}&\ \forall k\label{eq:pgfpmfrelation}\\
\mathbb{E}[X]&=\left[\mathcal{P}^{(1)}_{X}(s)\right]_{s=1}, &
\mathrm{Var}[X]&=\left[\mathcal{P}^{(2)}_{X}(s)\right]_{s=1}+\mathbb{E}\left[X\right]-\left(\mathbb{E}\left[X\right]\right)^2\label{eq:variance}.
\end{align}
Therefore, we get the following result.
{See Appendix \ref{prooftheorem1} for the proofs of the following results.} 		

\vspace{-.1in}	
\begin{theorem}\label{theorem1}
The PGF of the number $S(r)$ of points  of $\Psi_{\mrm}$  inside $\bt_{2}(\ob,r)$ is
	\begin{align}\label{PGFS}
		&\mathcal{P}_{S(r)}(s)=\exp\left(-2 \pi \densityRoads \left(r-\int_{0}^{r}\frac{\exp (g(s,t))t}{\sqrt{r^2-t^2}}{} \dv t \right)\right),
	\end{align}
\vspace{-.1in}	
	where $g(s,t)$ is given in \eqref{g}. 	\end{theorem}



\begin{corollary}\label{cor1.a1}
The PMF of $S(r)$ is given by
	\begin{align}
	\mathbb{P}[S(r)=n]&=\frac{1}{n!}
		{\mathcal{P}_{S(r)}(0)}\ 
		\bellf{}\left(f_{\mrm}^{(1)}(r),\cdots, f_{\mrm}^{(n)}(r)\right),\label{PSk}\\
	\text{with\ \ \  }	{f_{\mrm}^{(k)}(r)}&=2 \pi \densityRoads \int_{0}^{r}\frac{\exp\left(g(0,t)\right)}{\sqrt{r^2-t^2}}\ \bellf{}\left(g^{(1)}(0,t),\cdots, g^{(k)}(0,t)\right)t \dv t, \label{fsk}\hspace{1in}
	\end{align}
 and $g^{(k)}(0,t)$ can be evaluated from \eqref{g(0,rho)}. 
\end{corollary}

\vspace{-.2in}	
\begin{corollary}\label{cor1.1}
	The  mean and variance of $S(r)$ is 
	\begin{align}
		\mathbb{E}\left[S(r)\right]&={\left[{\mathcal{P}_{S(r)}^{(1)}(s)}\right]_{s=1}}=\densityPTS\pi r^2,\ 	\label{expecteds}\\
	\label{variance_s}
		\mathrm{Var}(S(r))&= 
		\begin{dcases}
			\textstyle\densityPTS\pi r^2+2\pi{\densityRoads}\left[\frac{32}{3}\left(a\lambdaPT\lambda_{\drm}\right)^2r^3+8\lambdaPT\lambda_{\drm}^2\left(\frac{2}{3}a r^3-\frac{1}{16}\pi r^4\right)\right],& \, a>r\\
			\textstyle\densityPTS\pi r^{2}+\left(\frac{8\lambdaPT\lambda_{\drm}a}{3}\right)^{2}\pi	{\densityRoads}{r^{3}}+4\pi{\densityRoads}\lambdaPT\lambda_{\drm}^{2}\left(r^{3}\left(\frac{8a}{3}-\frac{\pi r}{4}\right)\right.\\
			\textstyle\left.+\sqrt{r^{2}\!-\!a^{2}}\left(-\frac{a^{3}}{3}-\frac{13ar^{2}}{6}\right)+\left(2a^{2}r^{2}+\frac{r^{4}}{2}\right)\sin^{-1}\left(\frac{\sqrt{r^{2}-a^{2}}}{r}\right)\right).            &  \,a<r,
		\end{dcases}
	\end{align}
\end{corollary}

\vspace{-.2in}	
\subsection{Distribution of number of points of $\vehPP$ in a set under Palm {distribution}}
We also present the distribution of number of points  inside a set under Palm {distribution} (conditioned on occurrence of a point at the origin \ie $\ob\in\Psi_{\mrm}$). Similar to previous section, we can compute the PMF and the mean of $\hat{S}(r)$ from its PGF. 
See Appendix \ref{prooftheorem2} for the proof of the following results. 
\begin{theorem}\label{theorem2}
	The PGF of the number $\hat{S}(r)$ of points  of $\Psi_{\mrm}\setminus\{\ob\}$ conditioned on $\ob\in\Psi_{\mrm}$, falling inside  $\bt_{2}(\ob,r)$ is
	\begin{align}\label{PGFcaps}
		&\mathcal{P}_{\hat{S}(r)}(s)=\mathcal{P}_{S(r)}(s)\exp\left(g(s,r)\right)a^{-1}\int_{0}^{a}e^{(s-1)\lambda_{\drm}\A_{1}(r,a,x)}{\dv x},
	\end{align}
	where $\mathcal{P}_{S(r)}(s)$ is presented in Theorem \ref{theorem1}, $g(s,\cdot)$ is given in  \eqref{g}.
\end{theorem}
\begin{corollary}\label{cor_vehicle}
	The PMF of $\hat{S}(r)$ is 
	\begin{align}\label{PMFhatS}
		\mathbb{P}\left[\hat{S}(r)=n\right]&=
		\frac1 {n!}
\sum_{k_1+k_2+k_3=n}\left[ {n \choose {k_1,k_2,k_3}}\prod_{1\leq t\leq 3}f^{(k_t)}_{t}(0,r)\right],\\
\text{where }
%
%
	f_{1}^{(k)}(0,r)&=\mathcal{P}_{S(r)}^{(k)}(0)=
	{\mathcal{P}_{S(r)}(0)}\ 
		\bellf{}\left(f_{\mrm}^{(1)}(0,r),\cdots, f_{\mrm}^{(k)}(0,r)\right)\label{eq:Pskderivatives}\\
	f_{2}^{(k)}(0,r)&=e^{g(0,r)}\ \bellf{}\left(g^{(1)}(0,r)...g^{(k)}(0,r)\right)\\
	f_{3}^{(k)}(0,r)&=\int_{0}^{a}\left(\lambda_{\drm}
		\A_{1}(r,a,x)\right)^{k}{a^{-1}}e^{-\lambda_{\drm}\A_{1}(r,a,x)}{\dv x},
	\end{align}
%
	with  $g(\cdot)$ is given in \eqref{g} and  $f_{\mrm}^{(k)}(0,r)$ in \eqref{fsk}.
\end{corollary}

\begin{corollary}\label{cor2}
	The expected value of $\hat{S}(r)$ is 
		$\mathbb{E}\left[\hat{S}(r)\right]
={\densityPTS\pi r^2}
	+2\lambdaPT m r+{\lambda_{\drm}}\left(2r-r^{2}/(2a)\right).$
\end{corollary}
In this section, we have presented several key properties of PLP-MCP. These properties are PGFL, density, LF under the reduced palm {distribution}, the PGF (and PMF) for the number of points falling inside ball $\bt_{2}(\ob,r)$ both under normal and Palm distributions.  In the next section, we introduce the vehicular communication network providing connectivity to the platooned vehicles and present the distribution for the length of {the} typical and the tagged chord of {the} typical cell, and $0$-cell respectively of a homogeneous $2$D PPP.  These distributions are essential to derive the load distribution on the typical and tagged BS.
\section{Vehicular Communication Network}
The complete vehicular communication network consists of vehicular traffic (as defined in Section \ref{sec:vehicletraffic}) overlaid with the BSs forming a cellular network. 
The role of the cellular BS network is to provide cellular V2I connectivity  to vehicular users. We model the {locations} of BSs as a $2$D PPP $\BSPP\equiv\left\{\y_{i},:\y_{i}\in\R^{2},\, \forall i\in \mathbb{N}\right\}$ with density ${\lambda_{\bs}}$ \cite{SGBook2022}. Each BS transmits with the same power. The user association is based on the maximum average received power from the BSs and 
each user is connected to its nearest BS.  
Fig. \ref{system_model} shows the complete vehicular communication network.

Furthermore, we assume that the BS without any associated user will {stay silent and not create interference}. The active BSs point process $\Phi^{'}_{\bs}$ can be approximated as PPP with the active BS density $p_{\mathrm{on}}{\lambda_{\bs}}$ where $p_\mathrm{on}$ is the \textit{{active probability}} \cite{gupta2015potential}. 
Now, the signal to interference ratio {(SIR)} at the typical vehicle is given by
\begin{align}
	\SIR=\frac{ h_0 R^{-\alpha}}{\sum_{\y\in\Phi^{'}_{b}} h_{\y}\|\y\|^{-\alpha}},\label{eq:SIReq}
\end{align}
where 
$R$ denotes the distance of the nearest BS, $\alpha$ is path loss exponent, $h_{0}$ denotes the fading gain of the typical receiver link and $h_{\y}$ denotes the fading gains of { the} rest of the links. Further, we have assumed that the fading coefficients $h_{(\cdot)}$ are exponentially distributed with unit mean.

Due to the considered association policy, the serving region  of each BS is its Voronoi cell. For a BS located at $\y$, its Voronoi cell {$\Vor_{\y}$} is 
\begin{align}
	\Vor_{\y}=\left\{\x\in\R^{2}:\y={\arg\min}_{\y_i\in \BSPP}\|\x-\y_{i}\|\right\}.
\end{align}
Let $\mathsf{X}_\bs$ denote the union of Voronoi edges.
The  users connected to any BS {constitute}  the load on {that} BS.  
\begin{figure}[ht!]	
		\centering
		{\small \bf(a)} {\includegraphics[width=.44\linewidth]{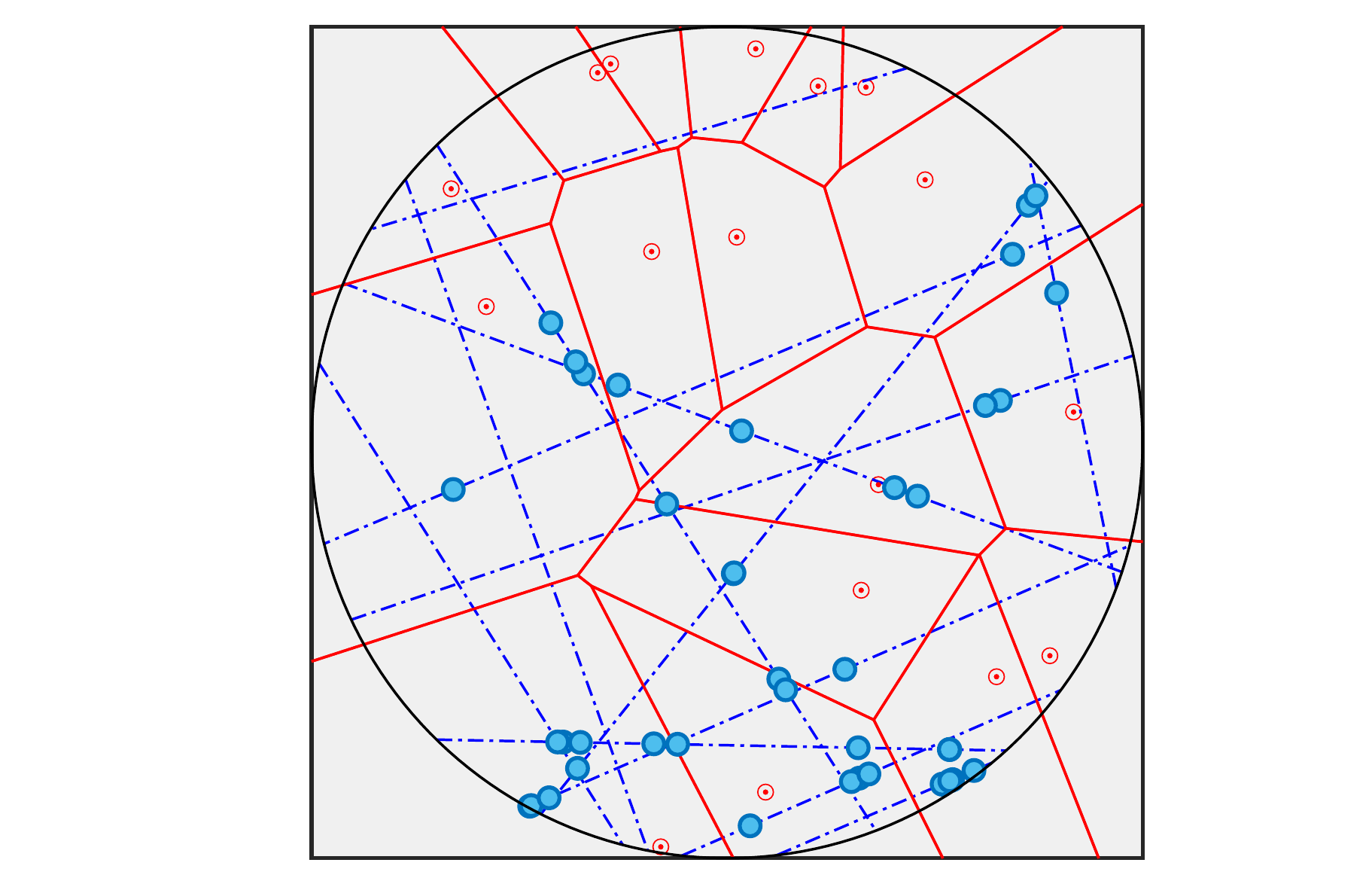} }
		{\small \bf(b)}{\includegraphics[width=.34\linewidth]{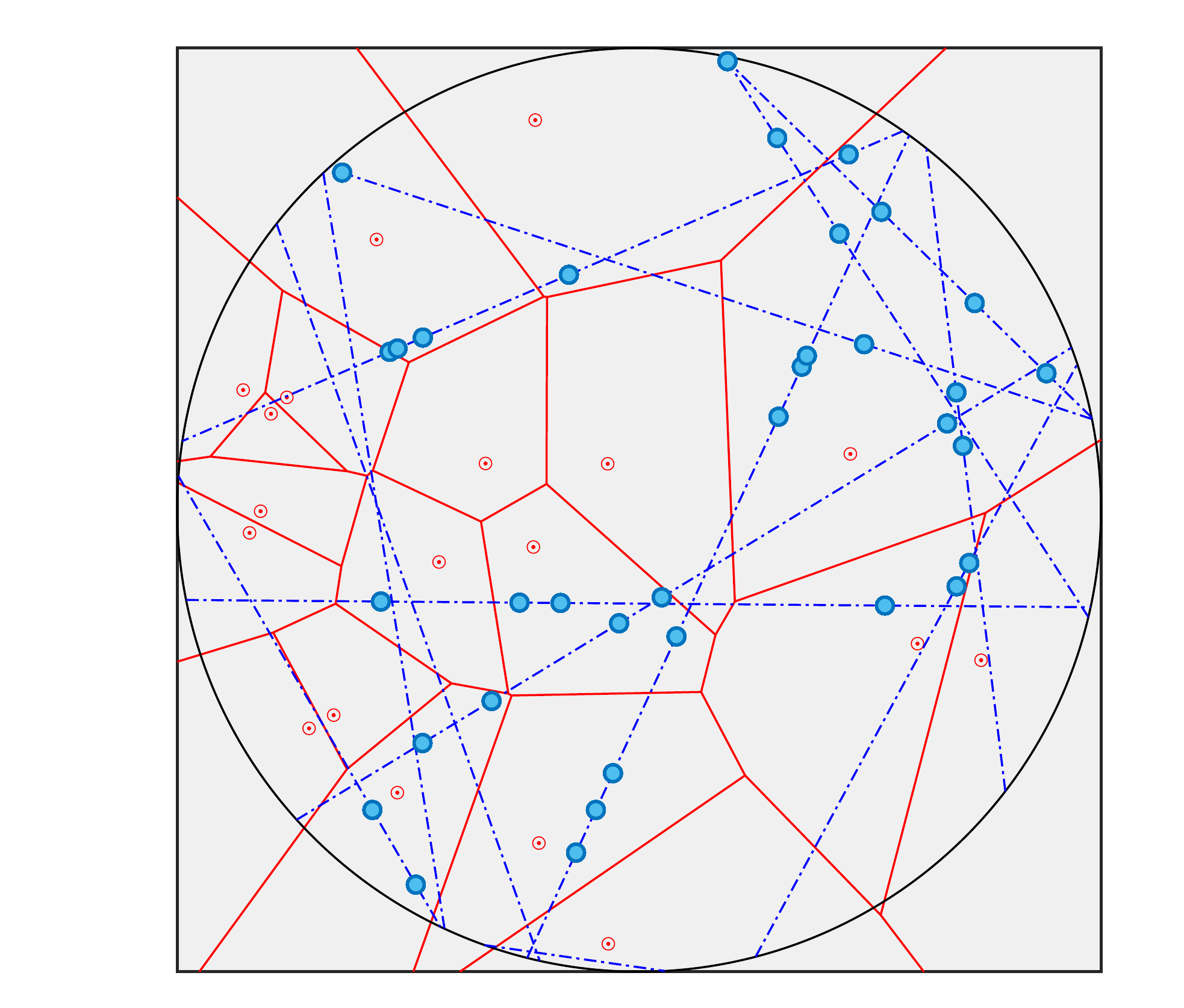} }
	\caption{(a) An illustration showing the complete vehicular network with platooned vehicles. The vehicles (shown as blue circles) on the road (shown as dotted line). Each BS (presented by {red} points) has an associated serving area (bounded by the {red} lines). As shown in the figure, the cluster movement of vehicles may assist them in connecting with the nearby vehicles for data, and content sharing. (b) A vehicular communication network with independently moving vehicles is shown for comparison. }
	\label{system_model}
	\adjustfigspace
\end{figure}

Let the typical Voronoi {cell} be $\Vor_{\typical}$. Its area $|\Vor_{\typical}|$ 
 is empirically distributed as a generalized gamma RV  \cite{tanemura2003statistical} 
	 with parameters 
$a_{1}=1.07950,b_{1}=3.03226$ and $c_{1}=3.31122$ \cite{ferenc2007size}. Hence, its PDF  is 
	\begin{align}\label{voronoi_area_with_lambda_b}
		g_{{|V_{\typical}|}}(v_{\typical})={\lambda_{\bs}}\tilde{g}_{X}\left({\lambda_{\bs}} v_{\typical};a_{1},b_{1},c_{1}\right).
	\end{align}
Similarly, the perimeter of $Z=L(V_{\typical})$ has the empirical  distribution 
\cite{tanemura2003statistical}
\begin{align}\label{perimeter}
	&p_{Z=z}(z)={\sqrt{{\lambda_{\bs}}}}/{4}\ \tilde{g}_{X}\left({\sqrt{{\lambda_{\bs}}}}z/{4};2.33609,2.97006,7.588060\right).
\end{align}



Now, let us consider {the typical vehicle at} the origin. The Voronoi cell   in which the origin falls is known as the {\em $0$-cell} \cite{chiu2013stochastic} \ie 
\begin{align}\label{taggedBS}
	\Vor_\ob&=\left\{
	\x\in \R^{2}:
	\underset{\y_i\in \BSPP}{\arg\min}  \|\x-\y_i \|=\underset{\y_i\in \BSPP}{\arg\min} \|\y_i \|\right\}. 
\end{align}
The BS $\y_\ob$ associated with {the $0$-cell} denotes the BS with which the typical vehicle at the origin connects and is termed the {\em  tagged BS}.

Owing to Ext-Slivnyak theorem (Lemma \ref{thm:slivE}), the load on the tagged BS consists of users of an independent copy of PLP-MCP falling {in} the tagged cell, plus a set of additional users falling on a part of the tagged line inside the tagged cell. As mentioned earlier, the tagged line or the road is the line on which the typical vehicle lies. The part of the tagged line falling inside the tagged cell is termed the {\em tagged chord}  $\mathsf{C}_\ob$ \ie $\mathsf{C}_\ob=\road_\ob\cap\Vor_\ob$. The tagged chord can also be defined as the chord of the tagged cell passing through the origin. {Since the} tagged chord's length $C_{\ob}$ plays an important role in the BS's load distribution, we derive its distribution  $f_{C_{\ob}}(c_{\ob})$  next. {While this specific result exists in \cite{chetlur2020load} within the context of the load distribution in a PLP-PPP, we derive it using a new approach that yields an easy-to-use expression that does not involve higher-order derivatives. We emphasize here that this chord length distribution is not our main contribution but just an important intermediate result that will facilitate further analysis.}
\subsection{Distribution of tagged chord length in  the Voronoi tessellation}
\begin{figure}[ht!]
\centering
{\bf \small (a)} \includegraphics[width=.5\textwidth,trim=0 5 0 80,clip]{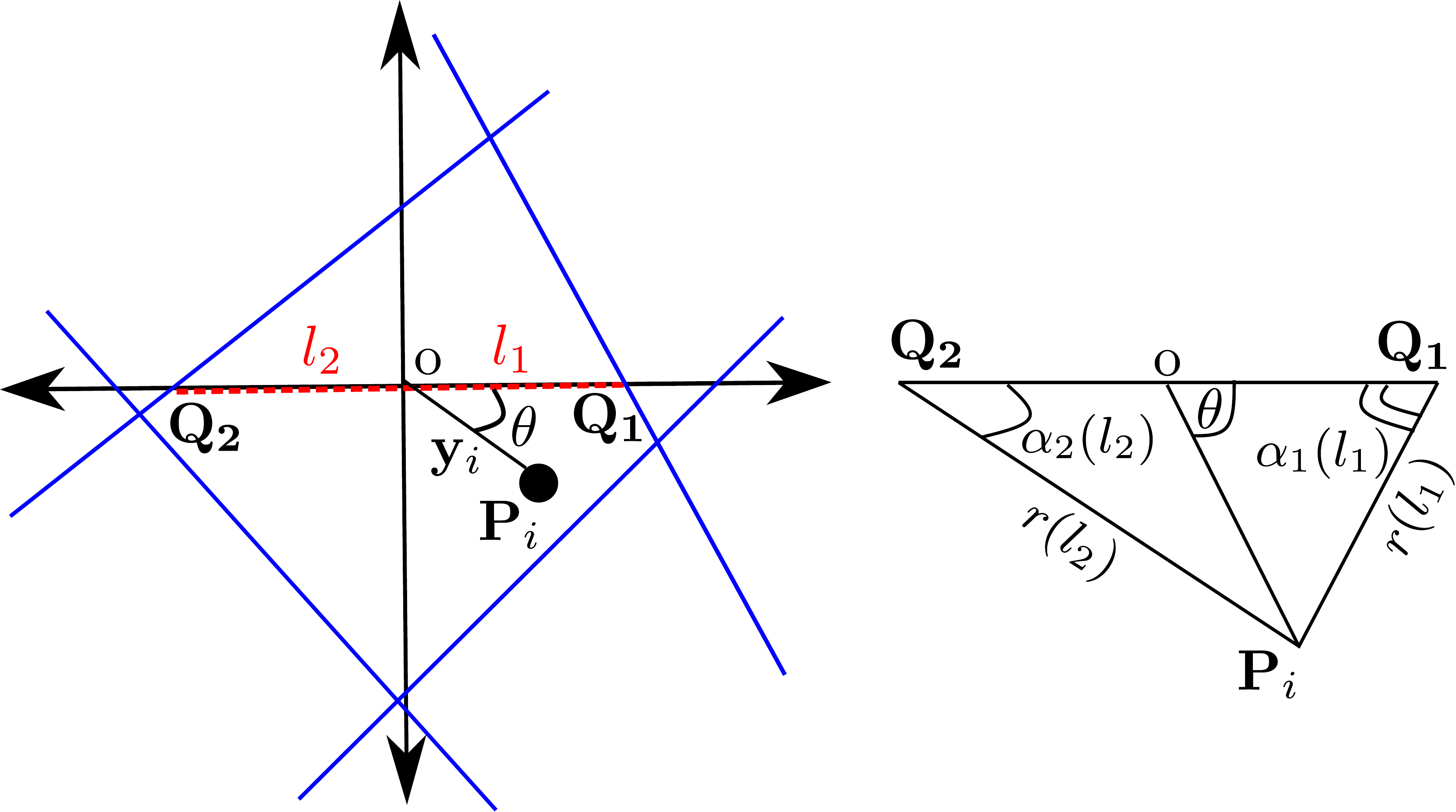}
			\centering
{\bf \small (b) }	\includegraphics[width=.35\textwidth]{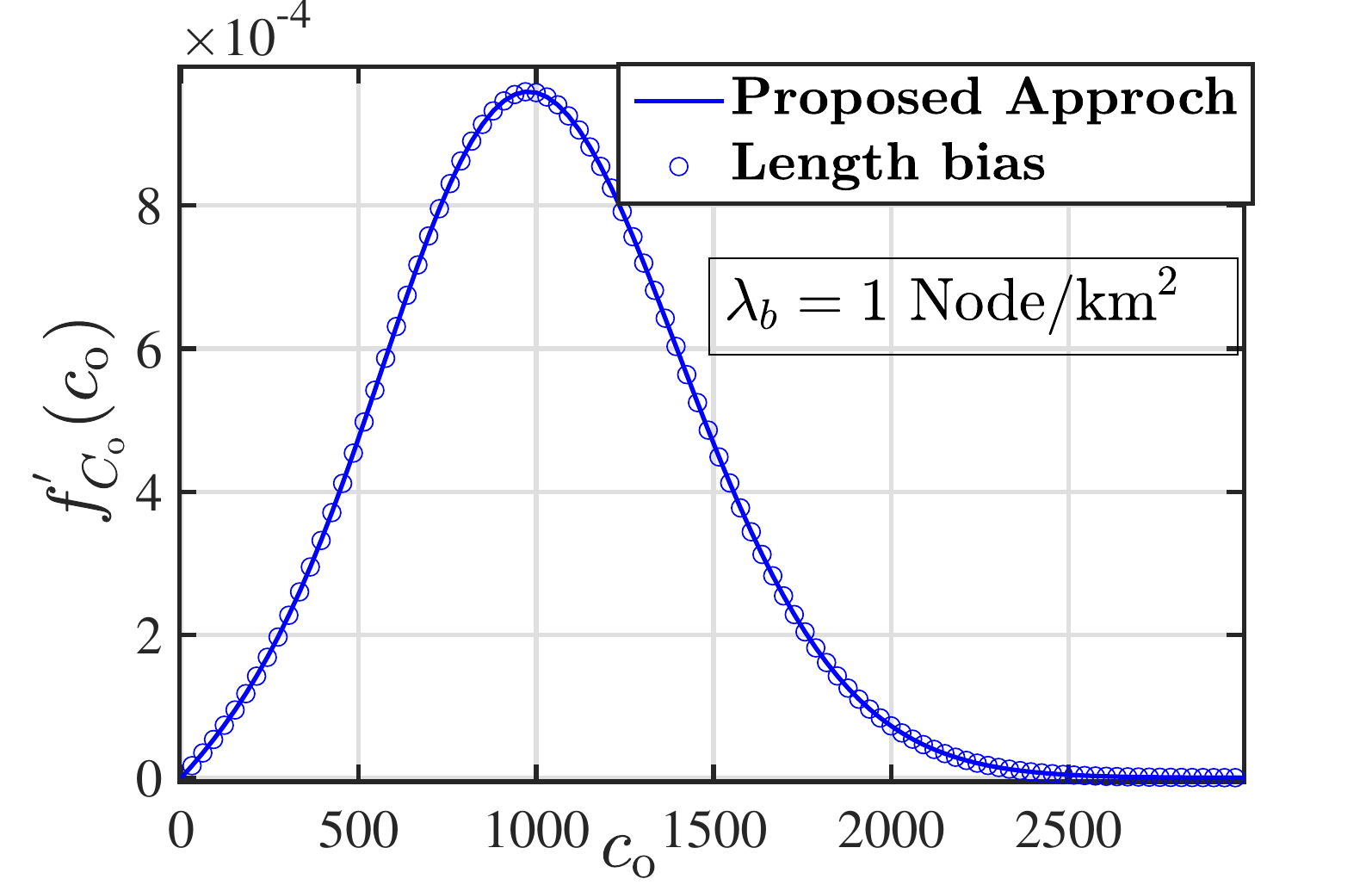}
	\caption{ (a) An illustration showing approach to find the PDF of the tagged chord length $\mathbf{Q}_1\mathbf{Q}_2$. Here, $\mathbf{P}_i$ denotes the BS.  (b)  The PDF obtained using the proposed method, along with the length bias result from \cite{chetlur2020load}.}
	\label{fig:taggedchordcomputation}
	\adjustfigspace
\end{figure}

We adapt an approach presented in \cite{gilbert1962random} to derive  the joint PDF of the length of the residual chords in both sides of the origin. Using the joint PDF, we derive the PDF of the length of the tagged chord.
%
We draw two lines from the origin in two opposite directions, (can be taken as positive and negative $\mathsf{x}$-axis without loss of generality). Further, the  points $\mathbf{Q}_1$ and $\mathbf{Q}_2$ where these two lines intersect $\mathsf{X}_\bs$,  are the two endpoints of tagged chord (as shown in Fig. \ref{fig:taggedchordcomputation}). 
 Let $l_1$ and $l_2$ be the distance of $\mathbf{Q}_1$ and $\mathbf{Q}_2$  from the origin. 
We first require the following result. 
\begin{lemma}\label{lemmad}
The radii $r_1$ and $r_2$ of two circles  $\bt_2((0,l_1),r_1)$ and $\bt_2((0,-l_2),r_2)$   such that they intersect at a point $\y=y\angle\theta$ (see Fig. \ref{fig:taggedchordcomputation}) are given as 
 $	r_1=r(l_{1})=\sqrt{l_{1}^{2}+y^{2}-2l_{1}y\cos\theta},\, r_2 = r(l_{2})=\sqrt{l_{2}^{2}+y^{2}+2l_{2}y\cos\theta}$ with angles  $\alpha_{1}(l_1)$ and $\alpha_{2}(l_2)$ as $
\cos\alpha_{1}(l_1)=\frac{l_{1}-y\cos\theta}{r(l_{1})},\, \cos\alpha_{2}(l_2)=\frac{l_{2}+y\,\cos\theta}{r(l_{2})}.$
The area of the union  
of these two 2D disks 
is given as
\begin{align}
		&\mathcal{V}(l_{1}+l_{2},r(l_{1}),r(l_{2}))=v_1(l_1)+v_2(l_2),\label{unionoftwocircle}	
\end{align}
where $v_i(l_i)=
r^2(l_{i})
\left(\pi-\alpha_{i}(l_i)+0.5\sin 2 \alpha_{1}(l_1)\right)$.
%
%
	Its partial derivative are
	\begin{align*}
	\frac{\partial \mathcal{V}}{\partial l_i}=v^{(1)}_{i}(l_{i}) 
		=2(l_i+y\cos \theta)(\pi-\alpha_{i}(l_i))+2 y \sin \theta.
	\end{align*}
\end{lemma}
%

Using the above result, we now derive the distribution of the tagged chord length which is given in the following theorem. See Appendix \ref{prooftaggechord} for the proof.
	\begin{theorem}\label{taggedchord}
		The joint PDF of the length of the two chord segments in the Voronoi tesselation is 
		\begin{align}\label{eq:taggedchordjointpdf}
			&{{f}_{L_1,L_2}(l_1,l_2)}=8{\lambda_{\bs}}^{3}\int_{0}^{\pi}\int_{0}^{\infty}e^{-{\lambda_{\bs}}\mathcal{V}(l_{1}+l_{2},r(l_{1}),r(l_{2}))}v^{(1)}_{1}(l_1)v^{(1)}_{2}(l_2)y \dv y \dv \theta,
		\end{align}
where $\mathcal{V}(l_{1}+l_{2},r(l_{1}),r(l_{2}))$ is given in \eqref{unionoftwocircle}.
The PDF of the length of the tagged chord  in the Voronoi tesselation is 
\begin{align}\label{tagged_chord}
	{f_{C_{\ob}}(c_\ob)}=\int_{0}^{c_{\ob}}{f_{L_1,L_2}}(c_{\ob},c_{\ob}-l_2)\dv l_2.
\end{align}
	\end{theorem}
	{For completeness, note that we derived an expression for the tagged chord length in \cite{chetlur2020load} using a length-biased sampling argument that provided the following expression,
	\begin{align}\label{taggedchordchteul}
		{	f_{C_{\ob}}(c_{\ob})}&=\frac{{c_{\ob}f_{C}(c_{\ob})}}{\mathbb{E}[C]}=\frac{4\sqrt{{\lambda_{\bs}}}}{\pi}c_{\ob}{f_{C}(c_{\ob})}
		\\
		&=\frac{4\sqrt{{\lambda_{\bs}}}}{\pi}c_{\ob} \frac{\pi}{2} {\lambda_{\bs}}^{\frac{3}{2}} \int_0^\pi \int_0^\infty\left[{\lambda_{\bs}} \left({ \mathcal{V}^{(1)}(c,y,r(c)) } \right)^2 -{ \mathcal{V}^{(2)}(c,y,r(c)) } \right] e^{- {\lambda_{\bs}} \mathcal{V}(c,y,r(c)) } y \dv y \dv \theta,\nonumber
	\end{align}
 which involved higher-order partial derivatives. The expression given above in Theorem \ref{taggedchord} is slightly simpler in that sense. Further, since this specific proof idea involving the joint distribution of two chords segments has not appeared in the literature, we decided to include it here.} Another advantage of the proposed approach is that it can also be extended to the case where the BS locations are distributed as a non-homogeneous PPP.
	
Equipped with the expressions of PDF of the tagged chord length and the number of vehicles in a set, we now analyze the vehicular communication networks in terms of the load per BS.

%

\section{Load Distribution in a Platooned Vehicular Communication Network}
In this section, we present the per-BS load distribution. The per-BS load in a communication system refers to {the} number of vehicles served by the  BS which is defined as the number of vehicles falling inside its Voronoi region. 
The distribution of per-BS load is an important performance metric as it critically affects the distribution of SINR, per-user available resources and finally the rate  in the following way. If a particular BS does not have any user associated with it, it may stay silent which reduces interference to the users of other BSs, and improves their SINR distribution.
Conversely, as the  time-frequency resources are split across the users associated with the serving BS, the  load on the tagged BS reduces the resources available to the typical user. 
As the rate distribution depends on the per-user resources and the SINR distribution,  the load on both the typical and the tagged BS plays a key role in the system's performance.  Hence, we will focus on distribution of the following important metrics:
\begin{enumerate}
\item $S_{\mrm}=\vehPP(\Vor_\typical)$: Load on the typical BS.
\item $M_{\mrm}=\vehPP(\Vor_{0})$: Load on the tagged BS.
\end{enumerate}
 The load distribution may help us decide the size of platoon and/or the number of vehicles in a platoon to improve performance. It may also provide us insights into the load distribution across the BSs that may help in optimizing the resource allocation, bandwidth sharing, and BS association. {This is especially important in the case of PTS that may exhibit larger disparity in the per-BS load, especially for smaller values of $a$. Since vehicles in a platoon drive in close proximity of each other, it is highly likely that vehicles in a given platoon are served by the same BS. This may lead to situations in which one BS serves multiple platoons and hence a large number of vehicles, whereas another BS does not serve any platoon and hence no vehicle.}
Therefore, it is crucial to understand the nature of load distribution on BSs. As mentioned already, we will assume that the BS remains silent (and hence does not create interference) if its load is zero.

 We will  look at an approximation ($\widetilde{S}_{\mrm}$ and  $\widetilde{M}_{\mrm}$ respectively) of these variable. To approximate the load in a Voronoi cell of  area  {$|V_{\typical}|$}, 
we will replace the cell with a 2D ball  of equal area, \ie the radius of this ball is $\sqrt{{|V_{\typical}|}/\pi}$ and instead compute the load in this ball. The PDF of the radius corresponding to the typical and tagged cell is respectively given as
\begin{align}\label{frt}
	&f_{R_\typical}(r_{\typical})=2 \pi r_{\typical} g_{\left|\Vor_{\typical}
		\right|}(\pi r_{\typical}^2).\\
\label{pdfr0}
	&f^{\ob}_{R_{\ob}}(r_{\ob})=2\pi r_{\ob}g_{|{\Vor}_\ob|}(\pi r_{\ob}^2)=2 \pi r_{\ob} \lambda_\bs \pi r_\ob^2\  g_{\left|\Vor_{\typical}
		\right|}(\pi r_{\ob}^2).
\end{align}
\subsection{Load distribution on the typical BS}	
	\begin{theorem}\label{theorem44}
		The PGF of the load $S_\mrm$ on typical Voronoi $V_{\typical}$ is (see Appendix \ref{prooftheorem44} for proof)
		\begin{align}\label{exact}
			\mathcal{P}_{S_{\mrm}}(s)&=\frac{\sqrt{{\lambda_{\bs}}}}{4}\int_{z=0}^{\infty}\exp\left(-{\densityRoads}z\left(1-\int_{0}^{\infty}\expS{g\left(s,.5{c}\right)}f_{C}(c)\dv c\right)\right) g_{|V_{\typical}|}\left(\frac{\sqrt{{\lambda_{\bs}}}}{4}z\right)\dv z,
		\end{align}
		where $g(s,\cdot)$ is given in \eqref{g}. The PMF of $S_{\mrm}$ is
		\begin{align*}
			\mathbb{P}\left[S^{}_{\mrm}=k\right]&=\frac1{k!}\left[{\mathcal{P}^{(k)}_{S^{}_{\mrm}}(0)}\right]=\frac{\sqrt{{\lambda_{\bs}}}}{4}\frac{1}{k!}\int_{z=0}^{\infty}	\mathcal{P}^{(k)}_{S^{}_{\mrm}\vert Z=z}(0)g_{|V_{\typical}|}\left(\frac{\sqrt{{\lambda_{\bs}}}}{4}z\right)\dv z,\\
			&=\frac{\sqrt{{\lambda_{\bs}}}}{4}\frac1{k!}\int_{0}^{\infty}{\mathcal{P}_{S_{\mrm}^{}|Z=z}(0)}\bellf{}\left(g_{\mrm}^{(1)}(0),\ldots,g_{\mrm}^{(k)}(0)\right) g_{|V_{\typical}|}\left({\sqrt{{\lambda_{\bs}}}}/{4}z\right)\dv z,\\
		\text{where }\hspace{.4in} \mathcal{P}_{S_{\mrm}^{}|{Z=z}}(0)&=	\exp\left(-{\densityRoads}z\left(1-\int_{0}^{\infty}\expS{g\left(0,{c}/{2}\right)}f_{C}(c)\dv c\right)\right)\nonumber\\
			g_{\mrm}^{(k)}(0)&={\densityRoads}z\int_{0}^{\infty}\exp\left(g\left(0,{c}/{2}\right)\right) \bellf{}\left(g^{(1)}\left(0,{c}/{2}\right),\ldots,g^{(k)}\left(0,{c}/{2}\right)\right)f_{C}(c)\dv c, 
		\end{align*}
	where  $g^{}(s,\cdot)$ is given in \eqref{g}.
		Further, $g^{(k)}(s,\cdot)$ provided in \eqref{g(0,rho)} can be evaluated at $s=0$.
	\end{theorem}

	Note that the mean load is equal to vehicular density times the mean size of the typical cell. Since the mean area of typical cell is $1/\lambda_\bs$, we get
	$$\expect{S_\mrm}=\densityPTS/\lambda_\bs. $$
%
%
%
	We can approximate $S_\mrm$ by $\widetilde{S}_\mrm=\Psi_{\mrm}(\bt_{2}(\ob,R_{\typical}))$. Note that conditioned on $R_{\typical}$, $\mathcal{P}_{\widetilde{S}_{\mrm}(R_\typical)|R_{\typical}=r_\typical}(s)=\mathcal{P}_{S(r_{\typical})}(s)$. Deconditioning  using the distribution of $f_{R_{\typical}}(r_{\typical})$, we get the following result.

\begin{theorem}\label{theorem3}
	The approximate PGF and PMF of the typical BS load are
	\begin{align} 
		\mathcal{P}_{\widetilde{S}_{\mrm}}(s)&=\int_{r_{\typical}=0}^{\infty}\mathcal{P}_{S(r_{\typical})}(s)\,f_{R_{\typical}}(r_{\typical})\dv r_{\typical}=2 \pi \int_{r_{\typical}=0}^{\infty}\mathcal{P}_{S(r_{\typical})}(s)\,r_{\typical} g_{{|V_{\typical}|}}(\pi r_{\typical}^2)\dv r_{\typical}.\\
		\mathbb{P}[\widetilde{S}_{\mrm}=k]&
		=2 \pi\int_{r_{\typical}=0}^{\infty}
 \prob{S(r_\typical)=k} r_{\typical} g_{{|V_{\typical}|}}(\pi r_{\typical}^2)\dv r_{\typical},
		\end{align}	
	where $\mathcal{P}_{S(\cdot)}(\cdot)$, and  $\prob{S(r_\typical)=k}$ are given in Theorem \ref{theorem1} and \eqref{PSk}, respectively. The PDF $g_{|V_{\typical}|}(\cdot)$ is given in \eqref{voronoi_area_with_lambda_b}.	
\end{theorem} 
\begin{corollary}\label{meanloadsm}
	The mean  of $\widetilde{S}_{\mrm}$ is
	$\mathbb{E}[\widetilde{S}_{\mrm}]=\densityPTS\pi \mathbb{E}\left[r_{\typical}^{2}\right]={\densityPTS}/{{\lambda_{\bs}}},$
	where $\mathbb{E}\left[r_{\typical}^2\right]=\frac{1}{\pi {\lambda_{\bs}}}$. Similarly, to find the variance of $\widetilde{S}_{\mrm}$, we need the second derivative of PGF conditioned on $s=1$ which is given as 
{	\[\lim_{s\rightarrow1}\mathcal{P}_{\widetilde{S}_{\mrm}}^{(2)}(s)=\int_{r=0}^{\infty}\left((F^{1}_{\mrm}(r))^{2}+F^{2}_{\mrm}(r)\right)f_{R_{\typical}}(r)\dv r\]
	\[F^{1}_{\mrm}(r)=2\pi{\densityRoads}\int_{t=0}^{r}\frac{\kappa(t,1)t \dv t}{\sqrt{r^{2}-t^{2}}},\,F^{2}_{\mrm}(r)=2\pi{\densityRoads}\int_{t=0}^{r}\frac{\left(\kappa^{2}(t,1)+\kappa(t,2) \right)t \dv t}{\sqrt{r^{2}-t^{2}}},\]}
	where $\kappa(t,k)$ is given in \eqref{kappaequation}. Using the second derivative, mean and variance equation present in \eqref{eq:variance}, we get the variance of $\widetilde{S}_{\mrm}$.
\end{corollary}
{
\begin{corollary}\label{thm:pon}
The active probability (or the {on} probability) of {the} typical BS is given as
\begin{align*}
p_\mathrm{on}=1-\prob{S_\mrm=0}=1- 2 \pi\int_{r_{\typical}=0}^{\infty}
\prob{S(r_\typical)=0} r_{\typical} g_{{|V_{\typical}|}}(\pi r_{\typical}^2)\dv r_{\typical} \\
\text{with } \prob{S(r_t)=0}= \exp\left(-2 \pi \densityRoads \left(r-\int_{0}^{r}\frac{\exp \left(g\left(0,t\right)\right)t}{\sqrt{r^2-t^2}}{} \dv t \right)\right).
\end{align*}
The off probability $p_\mathrm{off}=1-p_\mathrm{on}$.
\end{corollary}}
	\subsection{Load distribution on the tagged BS}
	In this section, we derive the approximate additional load  $\widetilde{M}_{\mrm}$ on the tagged cell. Unlike Theorem \ref{theorem44}, we will directly present the approximate load for {the} tagged cell. 
	Here, the load $\widetilde{M}_{\mrm}$ is equal to the sum of { the} number of vehicles on the tagged chord (of length $C_{\ob}$) and the number of vehicles falling inside a ball of radius $R_{\ob}$. From {Lemma \ref{thm:slivE}}, 
\begin{align*}
\widetilde{M}_\mrm\stackrel{(d)}{=}\Psi'_{\mrm}(\bt_{2}(\ob,R_{\ob}))+\Psi'_{\road_\ob}(\mathsf{C}_\ob)+\dauP'_{\bm{x}_\ob}(\mathsf{C}_\ob)
\end{align*}
where $\cdot'$ denotes the independent unconditional instances of the processes.
 Note that the total load counting the typical vehicle on tagged cell is $\widetilde{M}_{\mrm}+1$. 
\begin{theorem}\label{Theorem6}
	The PGF $\mathcal{P}_{\widetilde{M}_{\mrm}}(s)$ for the $\widetilde{M}_{\mrm}$ excluding the typical vehicle is
	\begin{align}\label{PGFexp}
		&\mathcal{P}_{\widetilde{M}_{\mrm}}(s)=\int_{c_{\ob}=0}^{\infty}\int_{r_{\ob}=0}^{\infty}\mathcal{P}_{\widetilde{M}_{\mrm}\vert R_{\ob},C_{\ob}}(s,r_{\ob},c_{\ob})f^{\ob}_{R_{\ob}}(r_{\ob})f_{C_{\ob}}(c_{\ob})\dv r_{\ob} \dv c_{\ob},
	\end{align}
	where,
	\begin{align*}
		&\mathcal{P}_{\widetilde{M}_{\mrm}\vert R_{\ob},C_{\ob}}\left(s,r_{\ob},c_{\ob}\right)=\mathcal{P}_{N_{\road_{\ob}}}\left(s,\frac{c_{\ob}}{2}\right)
\mathcal{P}_{S(r_{\ob})}(s)\int_{x_{\ob}=-a}^{a}\int_{x_c=-\frac{c_{\ob}}{2}}^{\frac{c_{\ob}}{2}}\frac{e^{(s-1)\lambda_{\drm}\A_{1}\left(\frac{c_{\ob}}{2},a,\left|x_{c}-x_{\ob}\right|\right)}}{c_{\ob}2 a}\dv x_c \dv x_{\ob},
	\end{align*} 
	where, $\mathcal{P}_{N_{\road_{\ob}}}\left(s,\frac{c_{\ob}}{2}\right)$ and $\mathcal{P}_{S(\cdot)}(s)$ is provided in \eqref{PGFMCP} and \eqref{PGFS}.
\end{theorem}
\begin{proof}
See Appendix \ref{proofTheorem6}.
\end{proof}
\begin{lemma}
	The PMF of $\widetilde{M}_{\mrm}$ is given as $\mathbb{P}\left[\widetilde{M}_{\mrm}=n\right]$ 
	\[=\frac{1}{{n!}}\left[{\int_{c_{\ob}=0}^{\infty}\int_{r_{\ob}=0}^{\infty}	\sum_{k_1+k_2+k_3=n}\left[ {n \choose {k_1,k_2,k_3}}\prod_{1\leq t\leq 3}h^{(k_t)}_{t}(0)\right]f^{\ob}_{R_{\ob}}(r_{\ob})f_{C_{\ob}}(c_{\ob})\dv r_{\ob} \dv c_{\ob}}\right],\]
	where $h_{1}^{(k)}(0,r_{\ob})=\mathcal{P}_{S}^{(k)}(0,r_{\ob})$ is obtained in \eqref{eq:Pskderivatives} 
	and
	\begin{align*}
		&h_{2}^{(k)}\left(0,{c_{\ob}}/{2}\right)=\exp\left(g\left(0,{c_{\ob}}/{2}\right)\right)\bellf{}\left(g^{(1)}\left(0,{c_{\ob}}/{2}\right),\ldots ,g^{(k)}\left(0,{c_{\ob}}/{2}\right)\right),\\
		&h_{3}^{(k)}\left(0,{c_{\ob}}/{2}\right)=\frac{1}{ac_{\ob}}\int_{x_{\ob}=0}^{a}\int_{x_c=-{c_{\ob}}/{2}}^{{c_{\ob}}/{2}}\left(\lambda_{\drm}\A_{1}\left({c_{\ob}}/{2},a,\left|x_{c}-x_{\ob}\right|\right)\right)^{k}{e^{(s-1)\lambda_{\drm}\A_{1}\left({c_{\ob}}/{2},a,|x_{c}-x_{\ob}|\right)}}{}\dv x_c \dv x_{\ob},
	\end{align*}
	and $g(0,\cdot)$ is given in \eqref{g}, $g^{(k)}(0,\cdot)$ is given in \eqref{g(0,rho)}.
\end{lemma}	
As the conditional PGF  $\mathcal{P}_{\widetilde{M}_{\mrm}\vert R_{\ob},C_{\ob}}\left(s,r_{\ob},c_{\ob}\right)$ is a product of three PGFs, the mean and variance of $\widetilde{M}_{\mrm}$ can be written as summation of the mean and variance of the three individual RVs.
\begin{corollary}\label{mean-variance-tagged-mcp}
	The  mean  of $\widetilde{M}_{\mrm}$ conditioned on $C_{\ob}$ is
	\begin{align*}
		\mathbb{E}\left[\widetilde{M}_{\mrm}\vert C_{\ob}=c_{\ob}\right]=
\frac{1.28\densityPTS}{ \densityBS}
+m\lambdaPT c_{\ob}+\frac{1}{2ac_{\ob}}\int_{x_{\ob}=-a}^{a}\int_{x_c=-\frac{c_{\ob}}{2}}^{\frac{c_{\ob}}{2}}\lambda_{\drm}\A_{1}\left(\frac{c_{\ob}}{2},a,|x_{c}-x_{\ob}|\right)\dv x_{c} \dv x_{\ob}.
	\end{align*}
We can further decondition using the PDF of $C_{\ob}$ as given in \eqref{tagged_chord}.
Similar to the variance of $\widetilde{S}_{\mrm}$, we first find the second derivative of the PGF of $\widetilde{M}_{\mrm}$ and then using the variance equation present in \eqref{eq:variance}, we find the variance of $\widetilde{M}_{\mrm}$.
\end{corollary}

\subsection{Load Distribution for the vehicular traffic under N-PTS}
		Since this paper also provides a comparative analysis of PTS with N-PTS, we also provide the load distribution for a vehicular communication network with N-PTS for completeness. We use the PGFs for load on the typical and the tagged cell presented in \cite{plpppp}
		 to derive the mean and the variance of load. Here, the vehicles  on each road form an independent $1$D PPP with density $\lambda$. Hence, the overall vehicular traffic $\Psi_{\prm}$ formed by taking the union of all the vehicles on all roads is a PLP-PPP, as discussed earlier already. Its density is 
$\densityITS=\pi {\densityRoads}\lambda$. Let $\mu_\prm$ denote the per-road vehicular density \ie $\mu_\prm=\lambda$.
%
%
The mean approximate load on the typical and tagged BS in vehicular traffic under N-PTS is given as follows. 
\begin{corollary}\label{Cor-mean-load-ppp}
	The mean and variance of the approximate load $\widetilde{S}_{\prm}$ on the typical BS is
	\begin{align*}
		&\mathbb{E}\left[\widetilde{S}_{\prm}\right]=\densityITS\pi \mathbb{E}\left[r_{\typical}^2\right]=\frac{\densityITS}{\densityBS},\,\mathrm{Var}\left[\widetilde{S}_{\prm}\right]=\left(\densityITS \pi\right)^2\mathbb{E}\left[r_{\typical}^4\right]+\frac{16}{3}\pi {\densityRoads}\lambda^2\mathbb{E}\left[r_{\typical}^{3}\right]+\frac{\densityITS}{\densityBS}-\left(\frac{\densityITS}{\densityBS}\right)^{2},\\
		&\text{where}\, \densityITS=\pi \densityRoads\lambda,\, \mathbb{E}\left[r_{\typical}^2\right]=\frac{1}{\pi \densityBS},\,\text{and}\,\, \mathbb{E}\left[r_{\typical}^{3}\right]=\int_{0}^{\infty}r_{\typical}^{3}f_{R_{\typical}}(r_{\typical})\dv r_{\typical}=\frac{\Gamma\left(\frac{c_1+1}{a_1}\right)}{b_{1}^{\frac{3}{2a_{1}}}(\pi{\lambda_{\bs}})^{3/2}\Gamma\left(\frac{c_1}{a_1}\right)}.
	\end{align*}

\end{corollary}

\begin{corollary}\label{mean-variance-tagged-ppp}
The mean of the approximate load $\widetilde{M}_{\prm}$ on the tagged BS is given as
\[	\expect{\widetilde{M}_{\prm}}=\expect{\mathbb{E}\left[\widetilde{M}_{\prm}\vert C_{\ob}=c_{\ob}\right]}=\densityITS \pi \mathbb{E}\left[r_{\ob}^2\right]+\lambda \expect{{C_{\ob}}}.\]
Similarly, we can find the variance of $\widetilde{M}_{\prm}$.
\end{corollary}

\section{Rate Coverage in a Platooned Vehicular Communication Network}
The rate coverage is defined as the probability that the rate achievable by {the} typical user is greater than a certain threshold \ie 
\[	\Rc(\tau)=\mathbb{P}(\mathcal{R}>\tau).\]
Assuming that the available bandwidth $B$ is equally shared by all user associated with the tagged BS, the achievable rate of typical receiver is given by
\[	\mathcal{R}={B}/{(1+\widetilde{M}_{\mrm})}\log_{2}\left(1+\SIR\right),\]
where {$\widetilde{M}_{(\cdot)}$} is the load on the tagged BS. Also note that {the} SIR depends on the active BS density which is further dependent on the load distribution on the typical cell. Hence, it is evident that the rate coverage  depends on the distributions of the user load on both the typical and  the tagged BS. 
Hence, the rate coverage is
	\begin{align}
		\Rc(\tau)&=\mathbb{P}\left({B}/{(1+\widetilde{M}_{\mrm})}\log_{2}\left(1+\SIR\right)>\tau\right)
		=\sum_{k=0}^{\infty}\mathbb{P}(\widetilde{M}_{\mrm}=k)\mathbb{P}\left(\SIR>2^{\frac{(k+1)\tau}{B}}-1\right)\label{final_expres}.
	\end{align}
{Here, $\mathbb{P}\left(\SIR>\tau\right)$ is the coverage probability of the typical user of a cellular network. For the channel and SIR model considered in \eqref{eq:SIReq}, the coverage is given as \cite{SGBook2022,gupta2015potential},  
 $\mathbb{P}(\SIR>\tau)$} 
\begin{align*}
		&=
		 2\pi {\lambda_{\bs}} \int_{0}^{\infty}\!\!r\exp\left(-{\lambda_{\bs}}\pi r^{2}\!-\!\mathrm{p}_{\mathrm{on}}\int_{r}^{\infty}\frac{2\pi {\lambda_{\bs}} \tau y \dv y}{\tau+ (\frac yr)^{\alpha}}\right)\!\!\dv r
=
		  \int_{0}^{\infty}\!\!\exp\left(-v\!-\!\mathrm{p}_{\mathrm{on}}\int_{v}^{\infty}\!\!\frac{   \dv u}{1+(\frac uv)^{\frac\alpha2}\tau^{-1}}\right)\!\!\dv v\\
		  &
		  =
		  \int_{0}^{\infty}\exp\left(-v\left(1+ \mathrm{p}_{\mathrm{on}}\int_{1}^
{\infty}\frac{   \dv t }{1+t^{\alpha/2}\tau^{-1}}\right)\right)\dv v=\frac1{1+ \mathrm{p}_{\mathrm{on}}\int_{1}^{\infty}\frac{   \dv t }{1+t^{\alpha/2}\tau^{-1}}},
	\end{align*}
	where the first two steps are due to the substitutions $\pi\lambda_\bs y^2=u$ and $u=vt$. Using \eqref{final_expres}, we get the following result.

\begin{theorem}\label{theorem9}
	The rate coverage  of {the} typical vehicular user in a vehicular communication network with platooned vehicles is 
	\begin{align}
	   &\Rc(\tau)=\sum_{k=0}^{\infty}\ {\prob{\widetilde{M}_{\mrm}=k}}{\left(1+ \mathrm{p}_{\mathrm{on}}\int_{1}^{\infty}\frac{   \dv t }{1+t^{\alpha/2}\gamma_k^{-1}}\right)}^{-1}, \label{rate_coverage}
	\end{align}	
where $\gamma_k=\left(2^{\frac{(k+1)\tau}{B}}-1\right)$ and $\mathrm{p}_{\mathrm{on}}$ is given in Corr. \ref{thm:pon}. \end{theorem}
Note that the rate coverage for a typical user in N-PTS can also be computed using \eqref{rate_coverage} by replacing $\widetilde{M}_{\mrm}$ and $\widetilde{S}_{\mrm}$ with $\widetilde{M}_{\prm}$ and $\widetilde{S}_{\prm}$, respectively. 
\section{Numerical Results}
In this section, we first present numerical results using the derived expressions. We will first verify the accuracy of the PMFs of $\widetilde{S}_{\mrm}$ and $\widetilde{M}_{\mrm}$ by comparing them with the exact simulation results.
We also discuss the impact of various parameters on the load distribution. After that, we will present a comparative analysis between PTS and N-PTS in terms of the energy efficiency, load imbalance and their impact on the rate coverage. In all our numerical results,  we use the following parameters unless stated otherwise. {The road density $\lambda_\roads=5/\pi$ \si{km^{-1}}, $\lambdaPT=1$  \si{platoons/km} $a=250$ \si{m}}. For fair comparison we have taken  $\lambda$ in N-PTS such that the total vehicular density $\densityITS$ in N-PTS is equal to $\densityPTS$.
\newcommand{\figsize}{.4\textwidth}
\begin{figure}[ht!]
	\centering
	{{\includegraphics[width=\figsize]{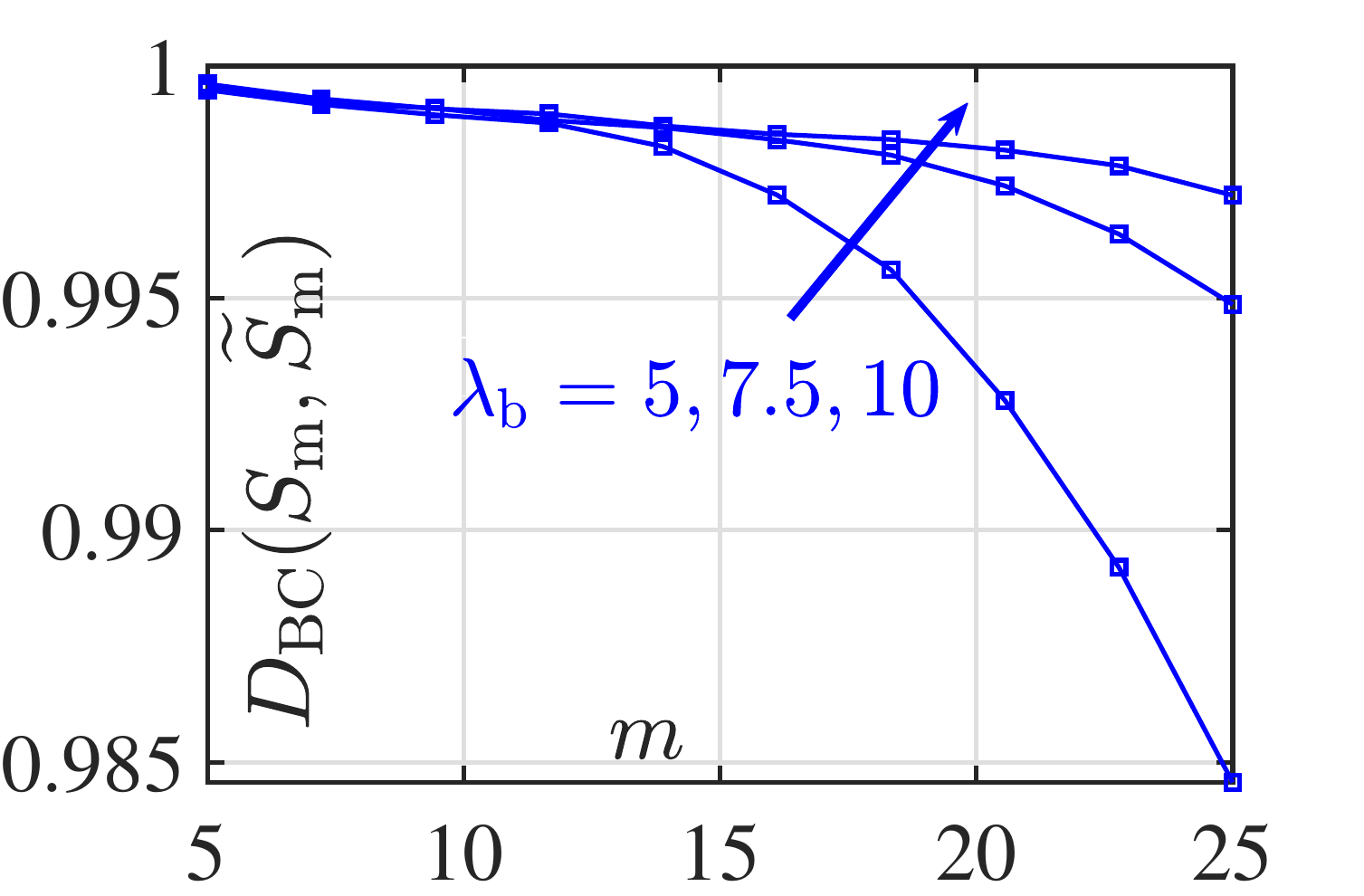} }}
	{{\includegraphics[width=\figsize]{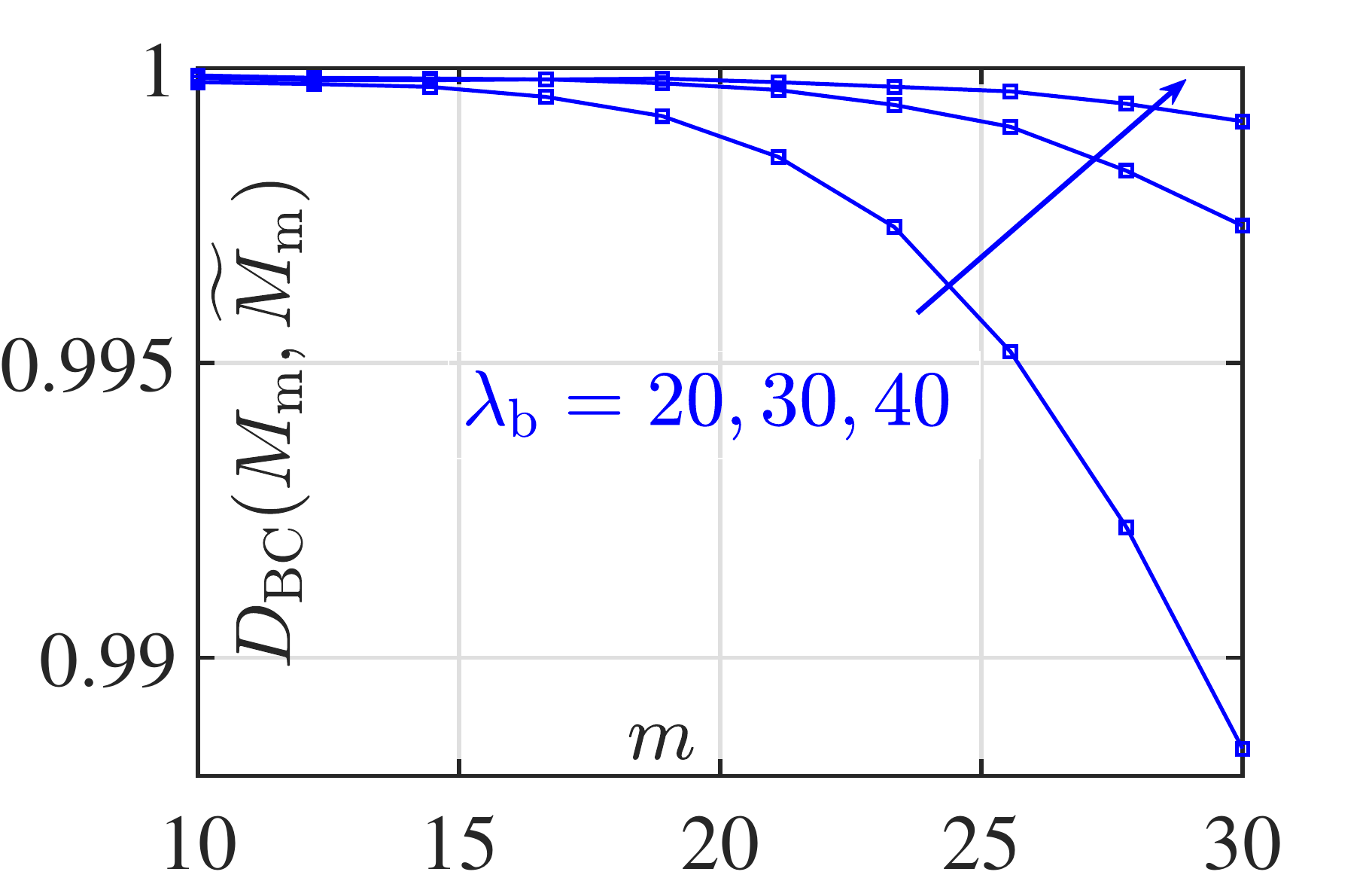} }}
	\caption{The BC for (a) the typical $D_{\mathrm{BC}}(S_{\mrm},\widetilde{S}_{\mrm})$ and (b) the tagged load $D_{\mathrm{BC}}(M_{\mrm},\widetilde{M}_{\mrm})$ for various values of $m$ and $\lambda_\bs$.
	 Here, { $\lambda_\roads=5/\pi$ \si{km}$^{-1}$} , $\lambdaPT=1$  \si{platoons/km} $a=250$ \si{m}. 
	 A value close to 1 implies that the approximation is accurate. }\label{bc}
\adjustfigspace \adjustfigspace
\end{figure}

\subsection{Validation}
To test the accuracy of  the derived distributions of the approximate load $\widetilde{S}_{\mrm}$ and $\widetilde{M}_{\mrm}$, 
 we evaluate the Bhattacharyya coefficient (BC) \cite{bhattacharyya1946measure} between the PMFs of the approximate load and the respective exact PMFs obtained using simulations. Note that for any two PMFs $p(\omega)$ and $q(\omega)$, the BC is defined as $D_{\mathrm{BC}}(p,q)=\sum\sqrt{p(x)q(x)}$. The BC $D_{\mathrm{BC}}(p,q)$, lies between $0$ to $1$, and a value close to 1 indicates good approximation. Fig. \ref{bc} {presents} the BC for the load on the typical $\left(D_{\mathrm{BC}}(S_{\mrm},\widetilde{S}_{\mrm})\right)$ and the tagged cell $\left(D_{\mathrm{BC}}(M_{\mrm},\widetilde{M}_{\mrm})\right)$.
{From this result, we notice that the approximation is remarkably close to the true result.} The approximation improves {further} with decrease in platoon size $m$ and increase in the BS density. 
\begin{figure}[ht!]
(a)
	{\includegraphics[width=\figsize]{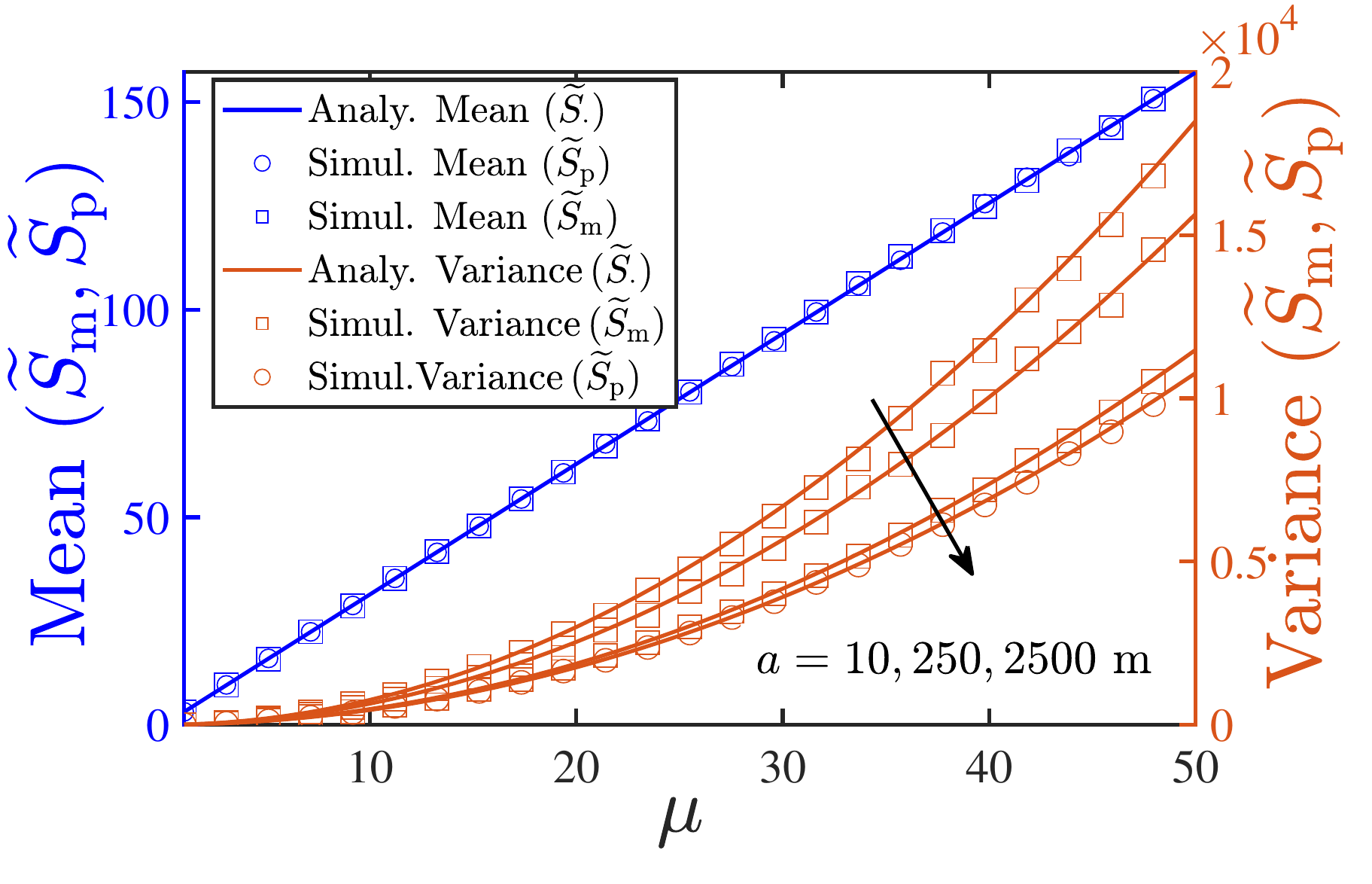} }
(b)
	{\includegraphics[width=\figsize]{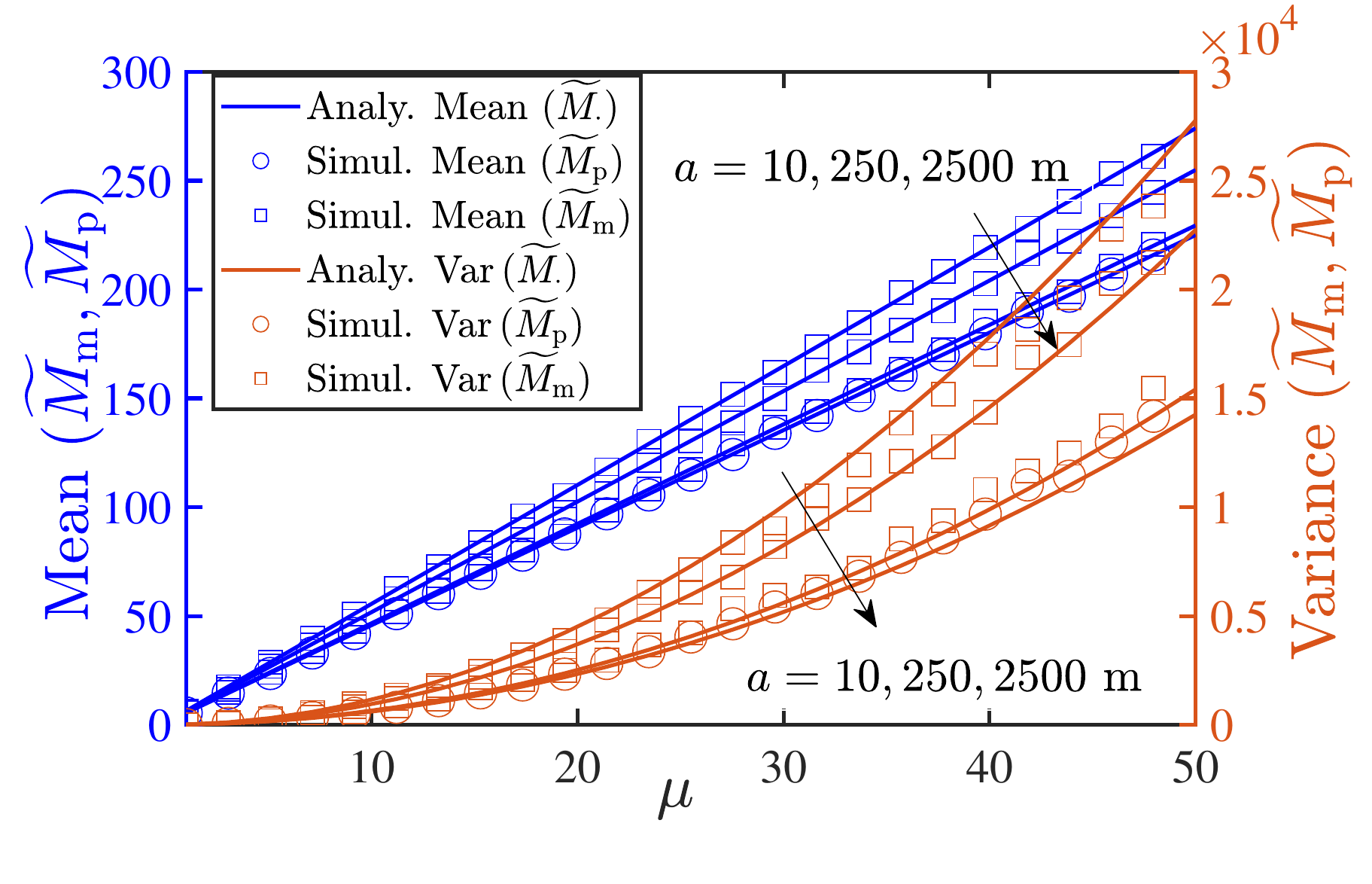} }
	\label{img2}
\caption{The mean and variance of the load on (a) typical ($\widetilde{S}_{\mrm},\,\widetilde{S}_{\prm}$) and (b) tagged ($\widetilde{M}_{\mrm},$  $\widetilde{M}_{\prm}$) cell with varying per-road vehicular density  $\mu_\mrm=\mu_\prm=\mu$. Here, for N-PTS, $\lambda$ is varied as $\lambda=\mu$  per \si{km}, while for PTS, $m$ is varied as $m=\mu/\lambdaPT$, while keeping $\lambdaPT=1$ \si{platoon/km}. Here, $\densityBS=5$ \si{BS/km^{2}}, $\lambda_\roads=5$ /\si{km}, and $a$ is in \si{meters}. As $a$ increases, the variance of $\widetilde{S}_{\mrm}$ and $\widetilde{M}_{\mrm}$ converges to variance of $\widetilde{S}_{\prm}$ and $\widetilde{M}_{\prm}$, respectively.}	\label{meanvarianceplot}
\adjustfigspace
\adjustfigspace
\end{figure}
\subsection{Mean and variance of the load on the typical cell }
Fig. \ref{meanvarianceplot}(a) shows the mean and variance for the approximate load on the typical cell with respect to per-road vehicular density $\mu_\mrm=m\lambdaPT$ for different values of the platoon radius $a$.
From Cor. \ref{meanloadsm}, the mean load on the typical cell depends linearly on density $\densityPTS=\mu_\mrm \densityRoads\pi$ but does not depend on $a$.  This is also evident from the numerical results.
It can be observed {further} that {the} variance grows quadratically with $\mu_\mrm$ {which is consistent with   \eqref{variance_s}.}
For small $a$, vehicles are concentrated close to the platoon centers because of which all the vehicles of a given platoon will very likely contribute to the load of a single BS. However, as we increase $a$, vehicles are more spread out, which decreases the variance of load on the typical BS.
We also present the respective metrics for {N-PTS}. Here, we keep $\mu_\prm=\mu_\mrm$ such that the mean load will be the same for {N-PTS} and PTS. Further, the variance for the PTS case is higher than N-PTS, and it becomes equal to N-PTS for very large $a$. This convergence is due to the fact that the MCP$(\lambda_\pt,m,a)$ converges to PPP($m\lambdaPT$), as $a\rightarrow\infty$
	\cite{pandeykth}.
\subsection{Mean and variance of the load on the tagged cell }
Fig. \ref{meanvarianceplot}(b) {presents} the mean and variance for the approximate load on the tagged cell with respect to per-road vehicular density $\mu_\mrm=m\lambdaPT$ for different values of  the platoon radius $a$. Here, mean and variance of $\widetilde{M}_{\mrm}$ are higher than those of $\widetilde{M}_{\prm}$. 
In PTS, the occurrence of the typical point adds points of the associated platoon in the load. Therefore, the mean and the variance of the load is higher in PTS compared to N-PTS. Further, as $a\rightarrow\infty$, the two scenarios become equivalent and the effect of the additional factor vanishes.

\begin{figure}[ht!]	
	\centering
	{\bf\small (a)}
	{\includegraphics[width=\figsize]{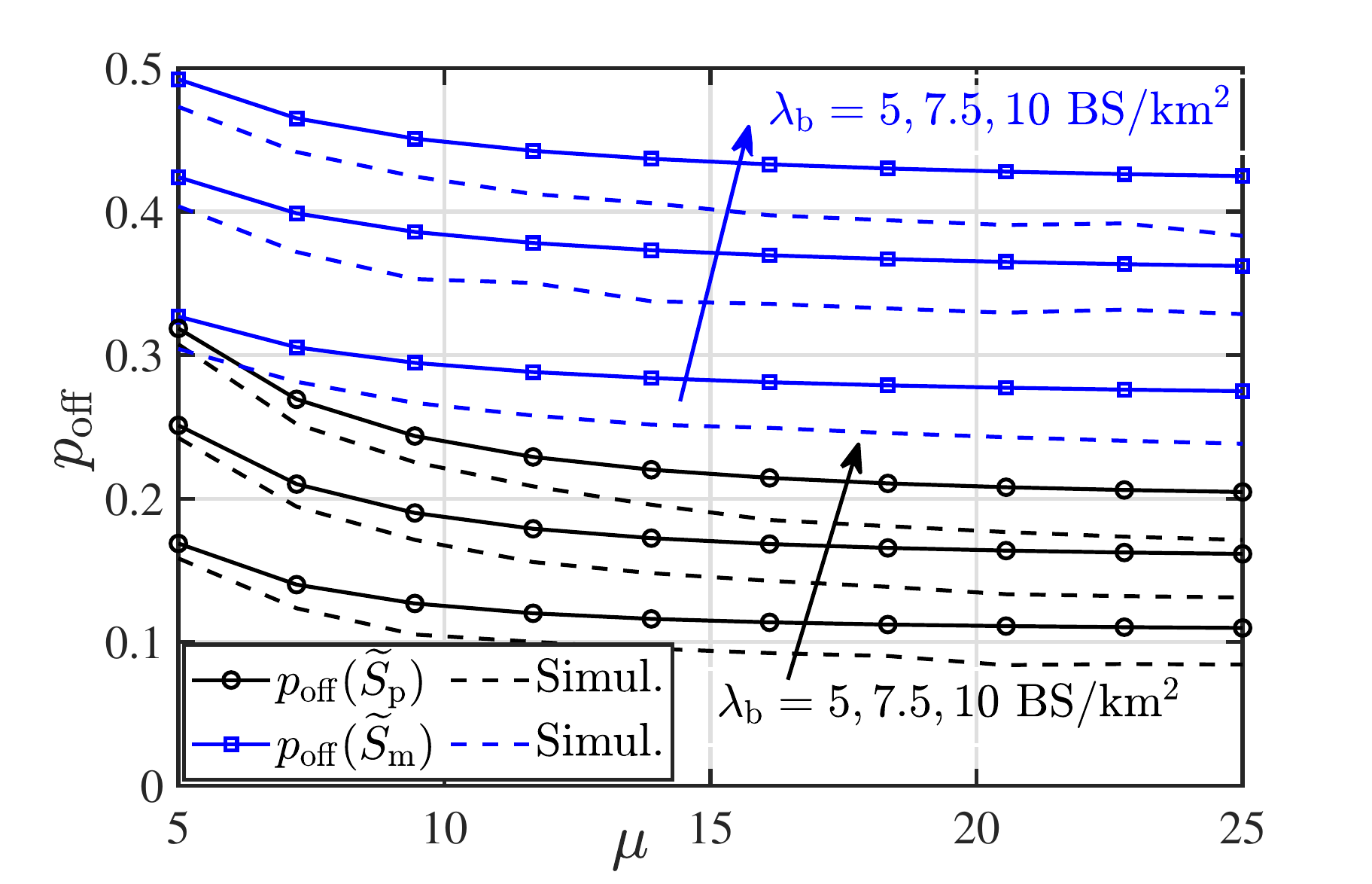} }
{\bf \small (b)}
	{\includegraphics[width=.43\textwidth]{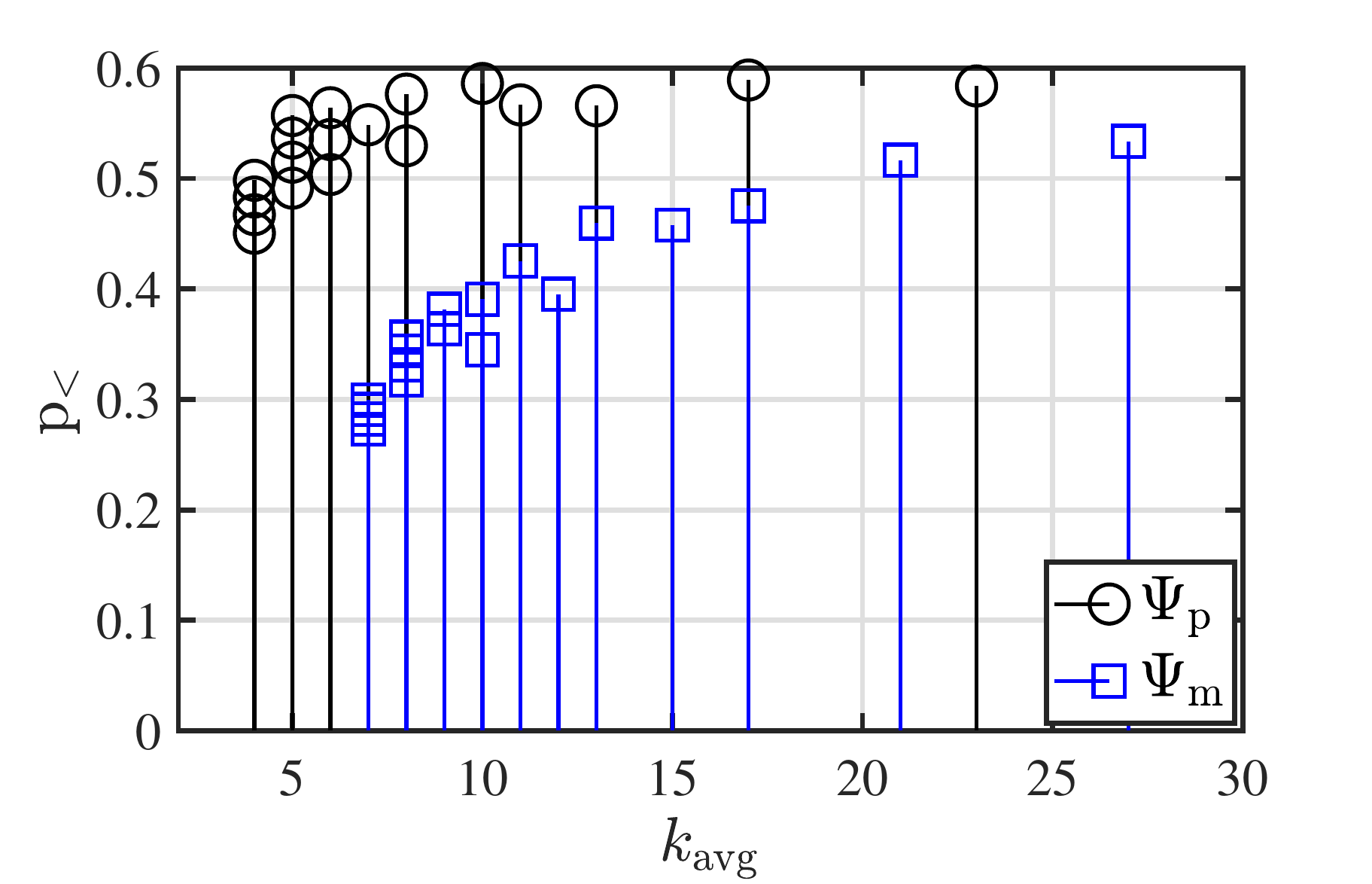} }
\caption{(a) The variation of $\off$ with respect to per-road vehicular density $\mu_\mrm=\mu_\prm=\mu$ for PTS and N-PTS. Here, for N-PTS $\lambda$ is varied as $\lambda=\mu$, {and for} PTS, $m$ is varied as $m=\mu/\lambdaPT$ while keeping $\lambdaPT=1$ \si{platoon/km}.
The fraction of BSs being switched off  is higher in PTS as compared to {N-PTS,} indicating a better energy efficiency.  (b) Variation of $\avg$ with $\kavg$. Here, $\mu=15$ \si{vehicles/km}. 
{N-PTS} has better underload probability. Here,  $a=250$ m in PTS. 
$\lambda_\bs$ is in  /\si{km^2}.
}
\label{fig:5}
\adjustfigspace \adjustfigspace
\end{figure}
\subsection{Impact of platooning on the energy efficiency of the typical cell}
To further understand the typical BS load, we will evaluate two additional metrics $\kavg$ and $\avg$. Here, $\kavg$ is defined as the mean load of the typical BS when it is active, \ie
 \begin{align*}
 	\kavg=\mathbb{E}\left[\widetilde{S}_{\cdot}\vert \widetilde{S}_{\cdot}>0 \right]
 =\expect{S_\cdot}/p_\mathrm{on}.
 \end{align*}
The second metric $\avg$ denotes the probability that the load on the typical active BS is less than the $\kavg$ \ie
\begin{align*}
	\avg=\mathbb{P}\left[\widetilde{S}_{\cdot}\le k_{\mathrm{avg}}\vert \widetilde{S}_{\cdot}>0\right].
\end{align*}
Note that $\avg$ represents the fraction of time the system is in a very safe operational regime. Fig \ref{fig:5}(a) presents the off probability $\off$ of the typical BS (which also represents the fraction of BSs staying silent) in PTS and N-PTS scenario with respect to per-road vehicular density $\mu=\mu_\mrm=\mu_\prm$. We observe that $\off$ is higher in PTS as compared to the N-PTS indicating that the energy consumption in PTS is less than N-PTS. 
 Fig. \ref{fig:5}(b) shows the variation of $\avg$ with active load $\kavg$ by varying $\lambda_{\bs}$ from $2-30$ \si{BS/km^{2}} while keeping the rest of the parameters fixed. As expected, the load $\kavg$ on active BSs decreases with $\lambda_\bs$. Further, $\kavg$ is high in PTS due to a lower fraction of BSs staying active as compared to N-PTS. Due to relatively higher load in PTS, safe-operating probability $\avg$ gets lower in PTS.  
\begin{figure}[ht!]
{\includegraphics[width=.321\linewidth]{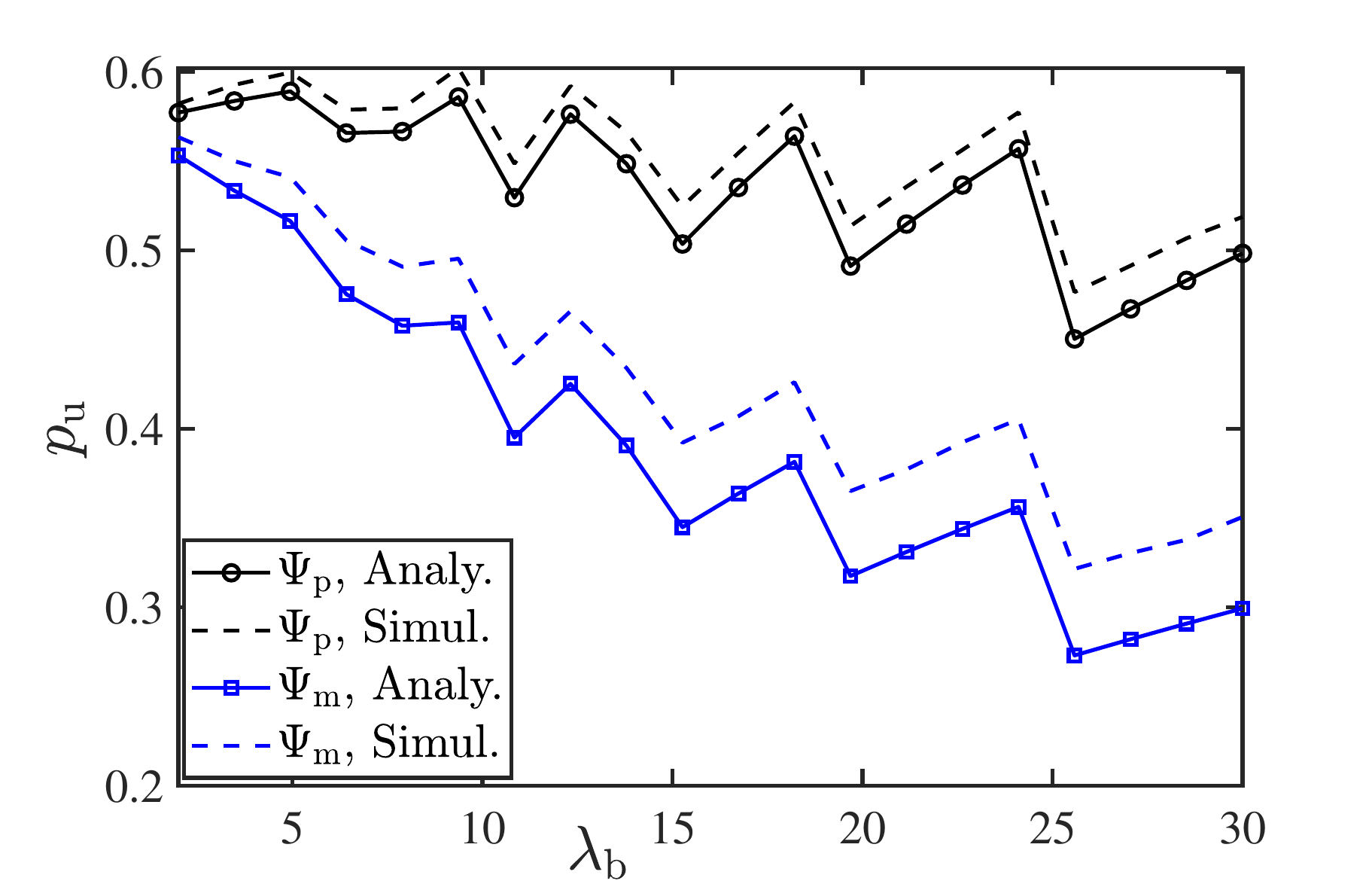} }
	\centering
	{\includegraphics[width=.321\linewidth]{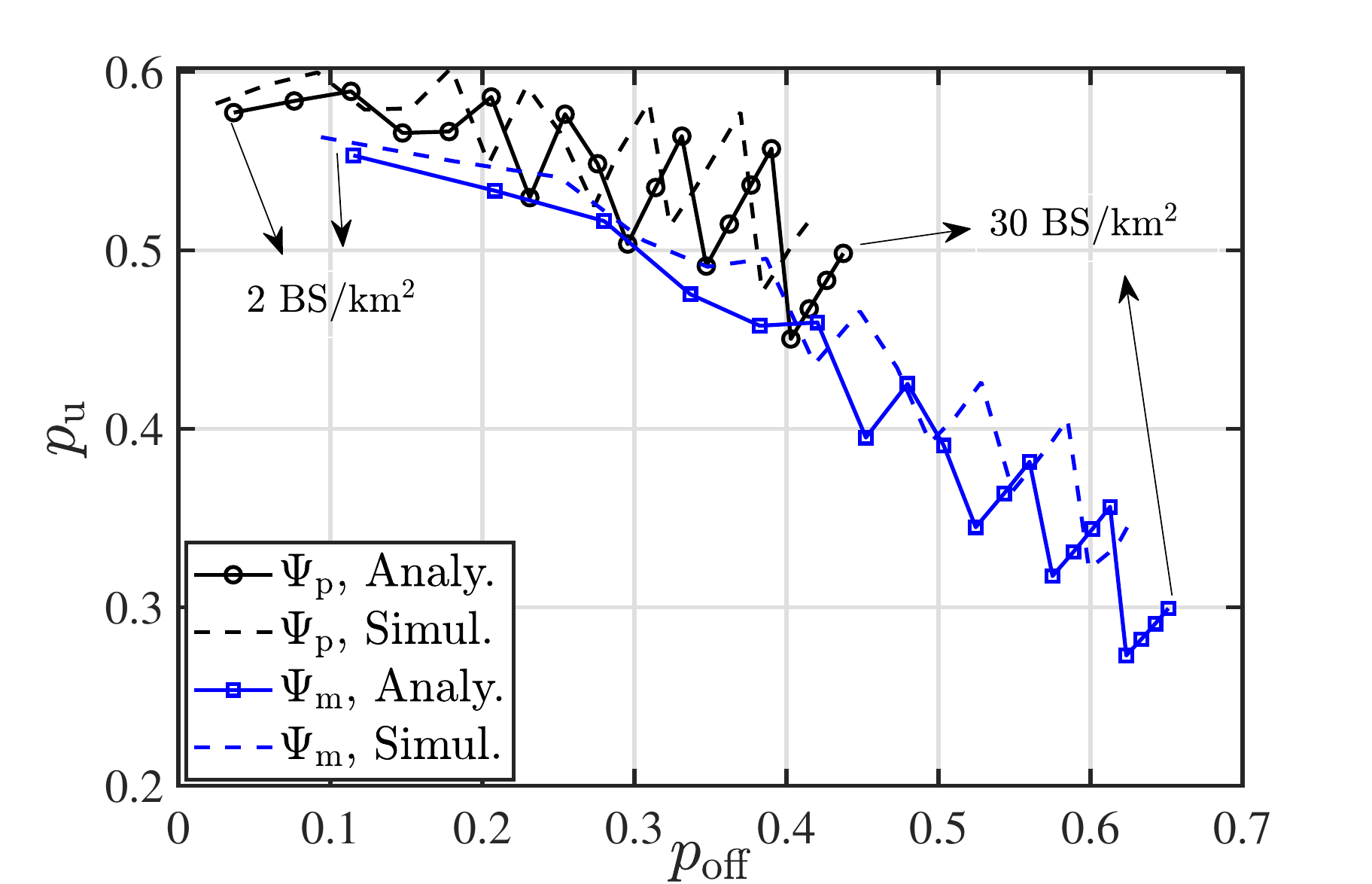} }
 {\includegraphics[width=.321\linewidth]{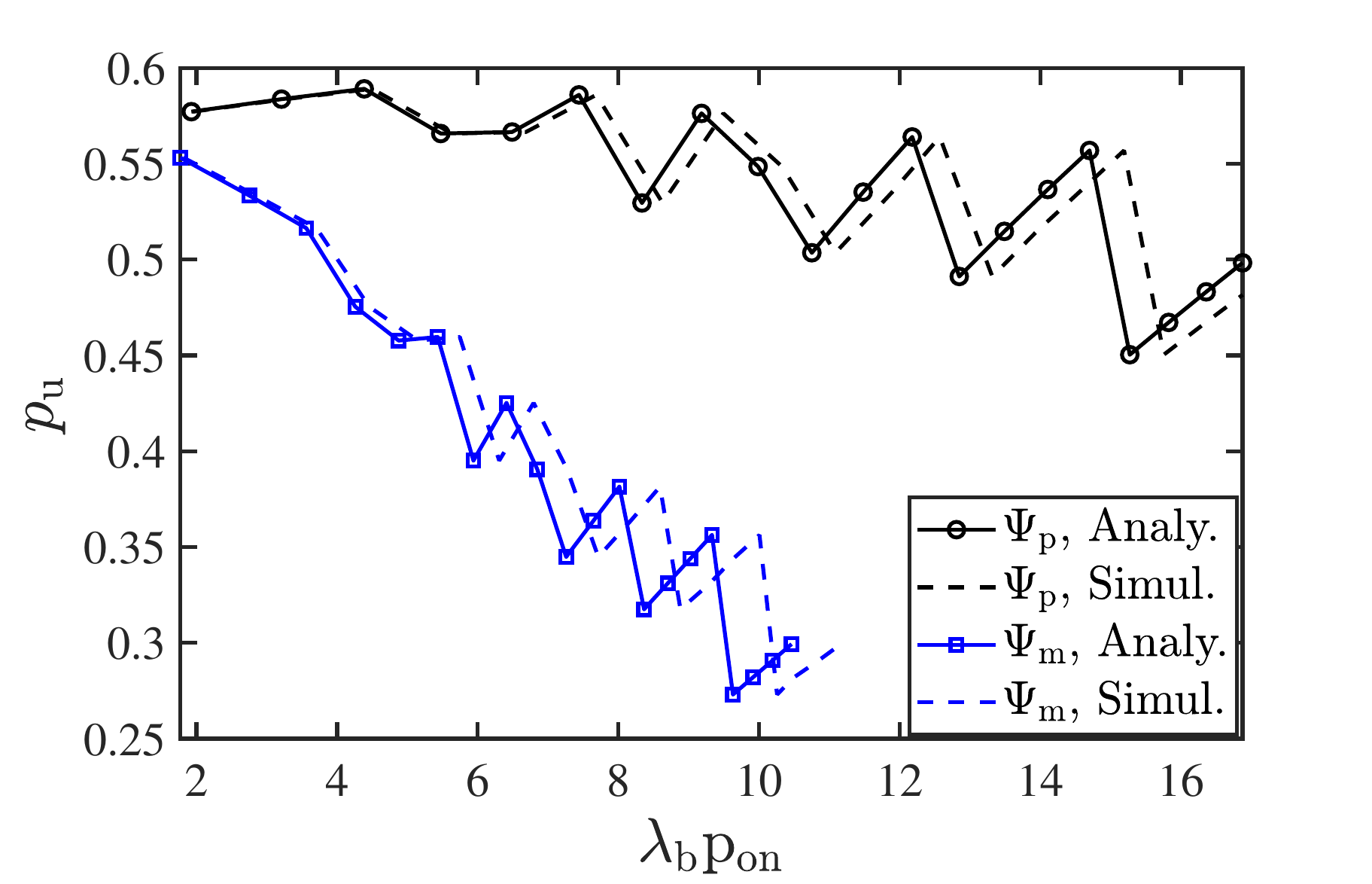} }\\\vspace{-.05in}
\bf \tiny (a) \hspace{2in} (b)\hspace{2in}	 (c)
\caption{ 
Variation of $\avg$ with respect to BS density, active BS density and off probability for {N-PTS} and PTS. 
Here, for {N-PTS} $\lambda$ is varied as $\lambda=\mu$, and for PTS, $m$ is varied as $m=\mu/\lambdaPT$ while keeping $\lambdaPT=1 \si{platoon/km}$.  $\mu=15$ \si{vehicles/km}. 
}
\adjustfigspace
\label{fig:6}
\end{figure}
Fig. \ref{fig:6}(a)-(c) presents variation of $\avg$ with respect to BS density, active BS density and off probability for N-PTS and PTS. Here, $\avg$ decreases with the densification of BSs which is intuitive. Small increments found in $\avg$ at some densities are due to the discrete nature of summation in the definition of $\avg$. {Elaborating further, first note that the $\widetilde{S}$ can take only integer values. When $\kavg$ decreases,  the number of individual PMF terms may not decrease if the change in $\kavg$ is fractional. However due to a decrease in the mean, the individual PMF terms increase, resulting in a net increase in $\avg$.}
In Fig. \ref{fig:6}(c), we compare $\avg$ between PTS and N-PTS by equating the off probability. We can see that for the same level of off probability, $\avg$ is almost the same in both the cases.
In Fig. \ref{fig:6}(c), we compare $\avg$ between PTS and {N-PTS} by equating the active BS density. Since the active probability is significantly lower in PTS, we can observe that at the same value of $\lambda_\active=\lambda_\bs$, $\avg$ in N-PTS is lower.
\begin{figure}[ht!]
		\centering
		{\small \bf (a)} \includegraphics[width=.29\linewidth]{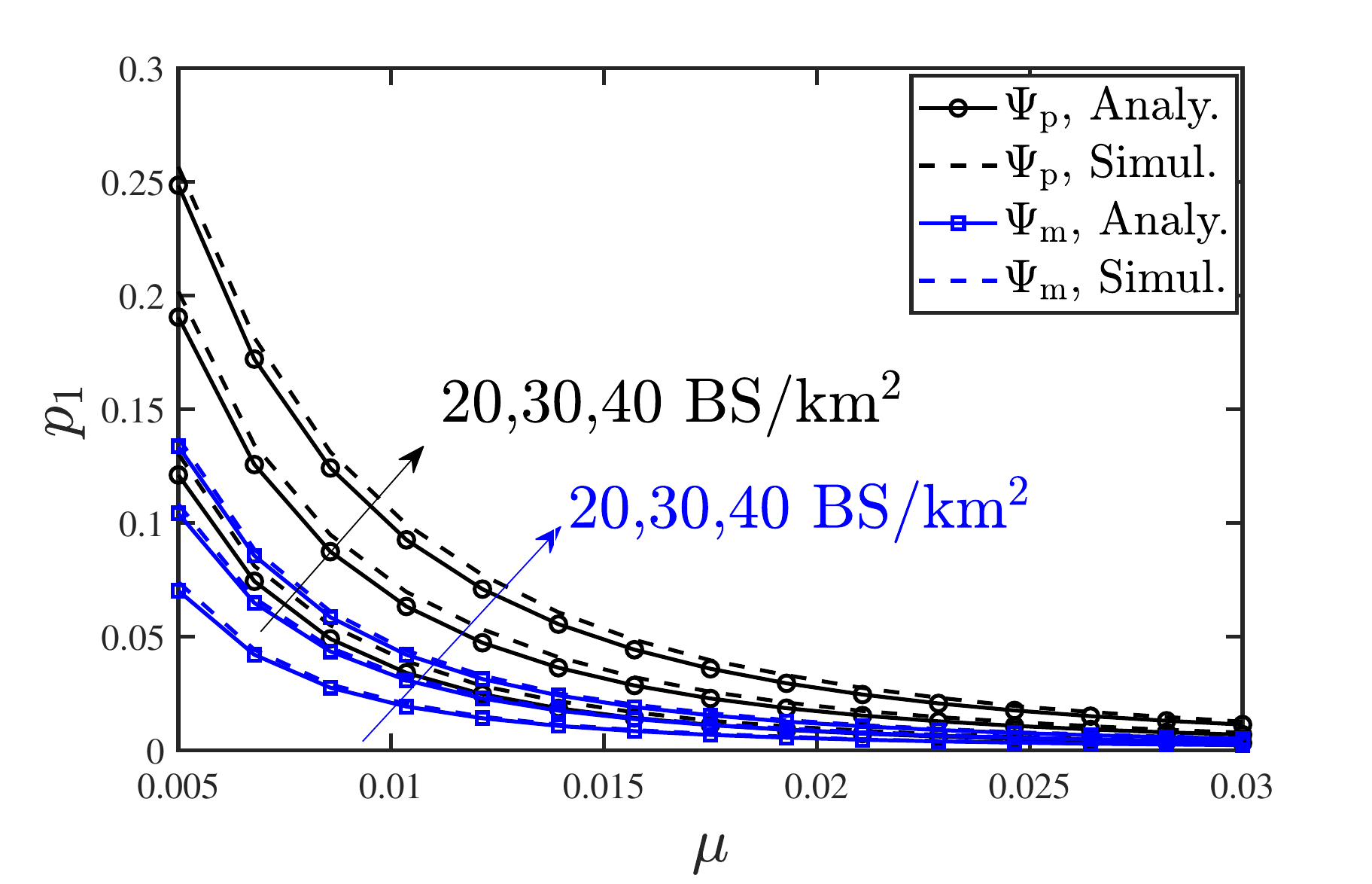} 
{\small \bf (b)}
	\includegraphics[width=.29\linewidth]{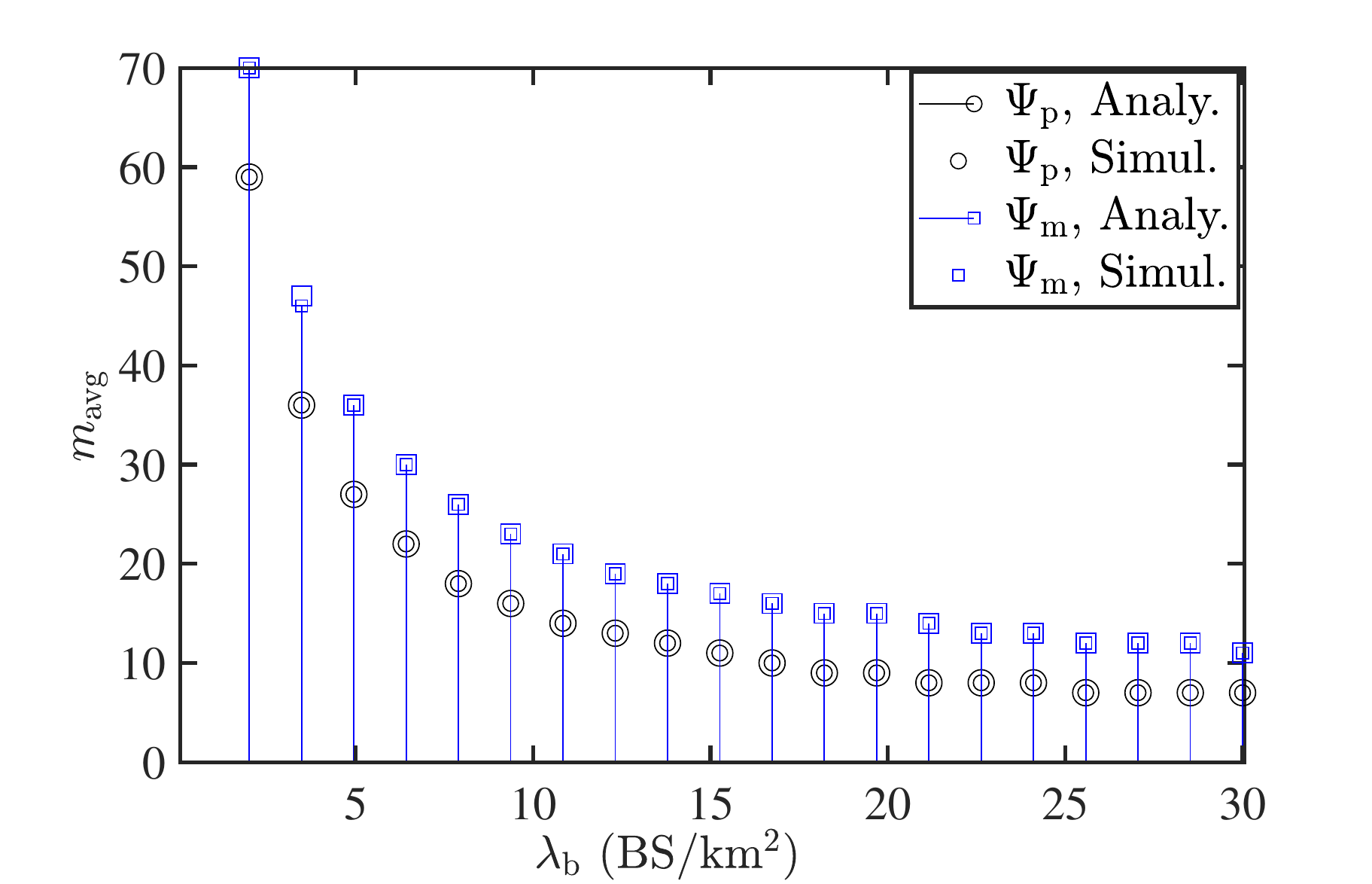} 
{\small \bf (c)}
	\includegraphics[width=.29\linewidth]{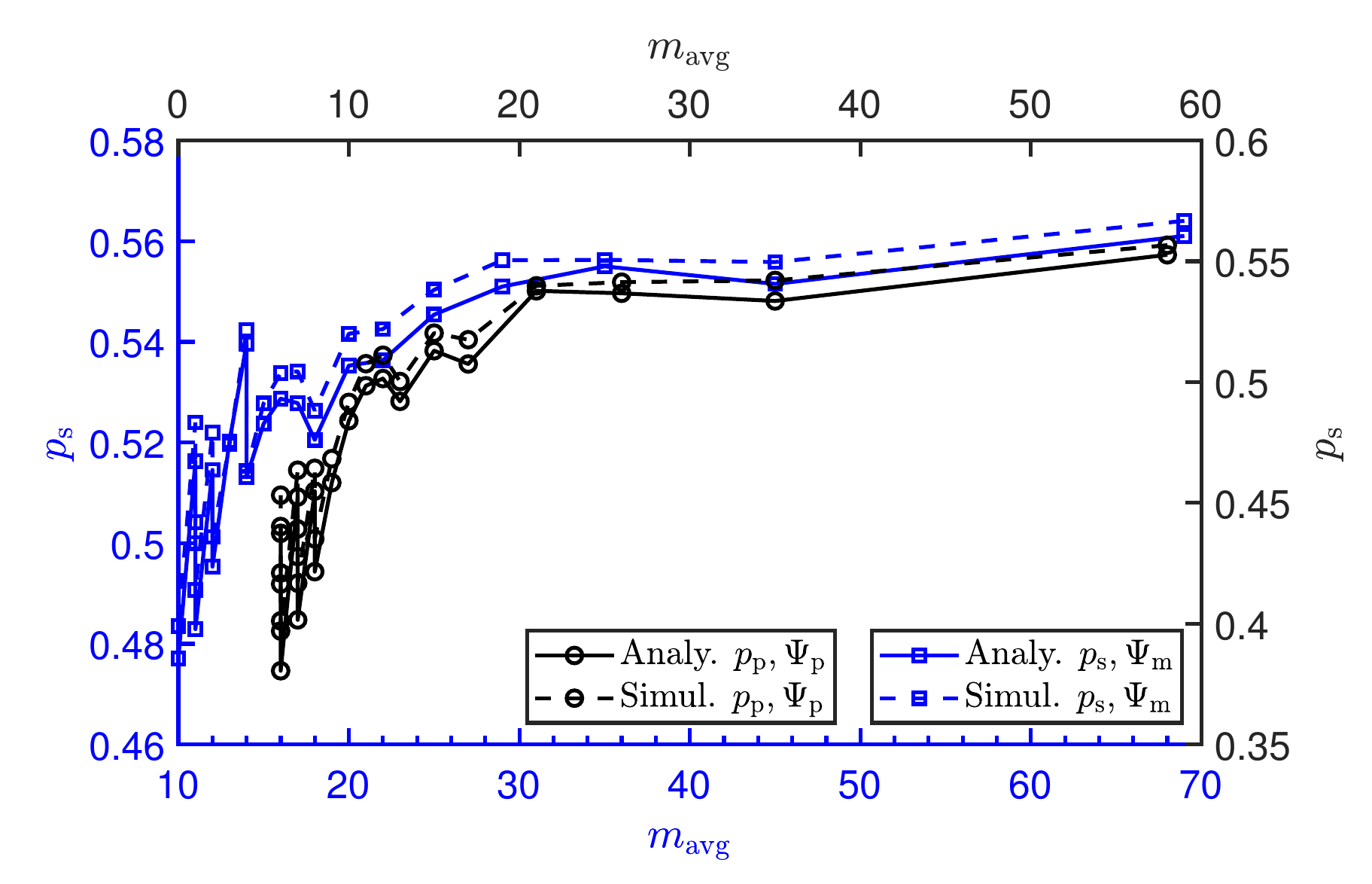} 
\caption{ (a) The variation of $\Pmin$ with respect to per-road vehicular density $\mu=\mu_\mrm=\mu_\prm$, (b) mean load $\mavg$ on the tagged cell  with $\lambda_\bs$, and (c) safe-operating probability $\Pavg$ with respect to $\mavg$. In (a), $\lambda_{\roads}=2\, \si{km^{-1}}$, for (b), and (c) the parameters are $\lambda_{\roads}=5/\pi$ \si{km^{-1}},  $\lambda=m\lambdaPT=15$ \si{vehicles/km}. }
\adjustfigspace
\label{fig:7}
\end{figure}

\subsection{Load balance in the tagged cell}
Note that the mean load is not the only criteria for comparing two systems. For instance, it may not be optimal from the energy utilization perspective to activate a BS to just serve a single vehicle. In order to understand the effect of load distribution, we define the following metrics: single user probability $\Pmin$, the average load on tagged cell $\mavg$, and tagged safe-operating probability $\Pavg$, as follows
\[
\Pmin=	\mathbb{P}\left[\widetilde{M}_{\cdot}=1\right],\,
m_{\mathrm{avg}}=\mathbb{E}\left[\widetilde{M}_{\cdot}\right]
,\,\Pavg=\mathbb{P}[\widetilde{M}_{\cdot}\le \mavg].
\] 
Note that, a high $\Pmin$ represents that many BSs in the system are severely underloaded. From Fig \ref{fig:7}(a), 
we can observe that  $\Pmin$ in lower in PTS. This indicates that it is more likely in PTS that a BS is not wasting its power to just serve a single user. Fig \ref{fig:7}(b) shows the variation of mean load in the tagged cell which decreases with $\lambda_\bs$. Unlike the typical cell, the mean load on the tagged cell differs in PTS and {N-PTS}. Fig. \ref{fig:7}(c) shows $\Pavg$ with respect to $\mavg$ using the data obtained from Fig. \ref{fig:7}(b). 
We can observe that $\Pavg$ is higher in N-PTS for the same value of $\mavg$. Together with Fig. \ref{meanvarianceplot}(b), which shows that the variance of the load on the tagged cell is higher in PTS, we can see that   the spread of load distribution is higher in PTS.  This means that the tagged BS may have to support a higher number of users in PTS compared to N-PTS.
\begin{figure}[ht!]%
		{\small \bf (a)}
		{\includegraphics[width=.44\linewidth]{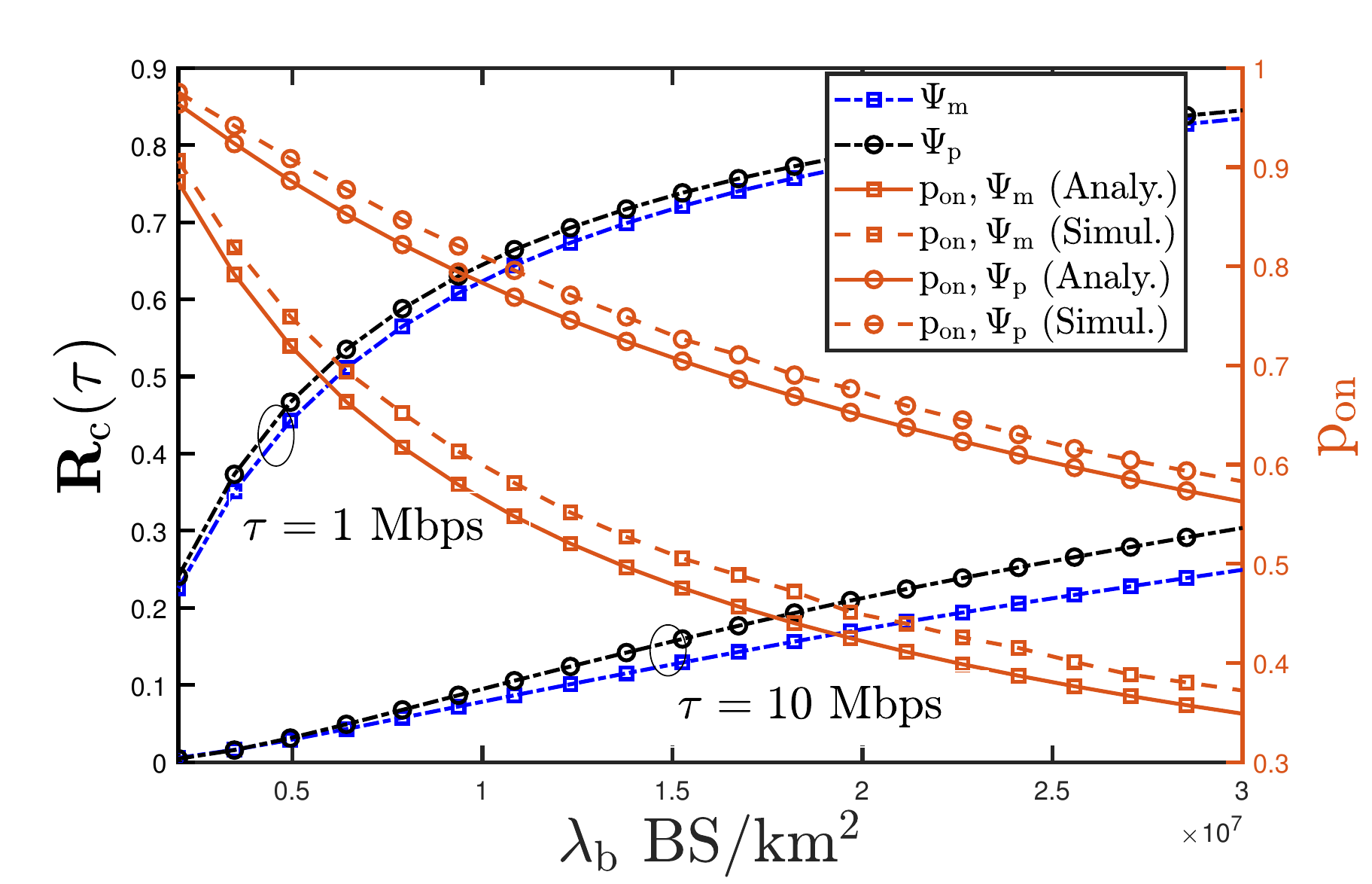} }
	{\small \bf (b)}
	{\includegraphics[width=.42\linewidth]{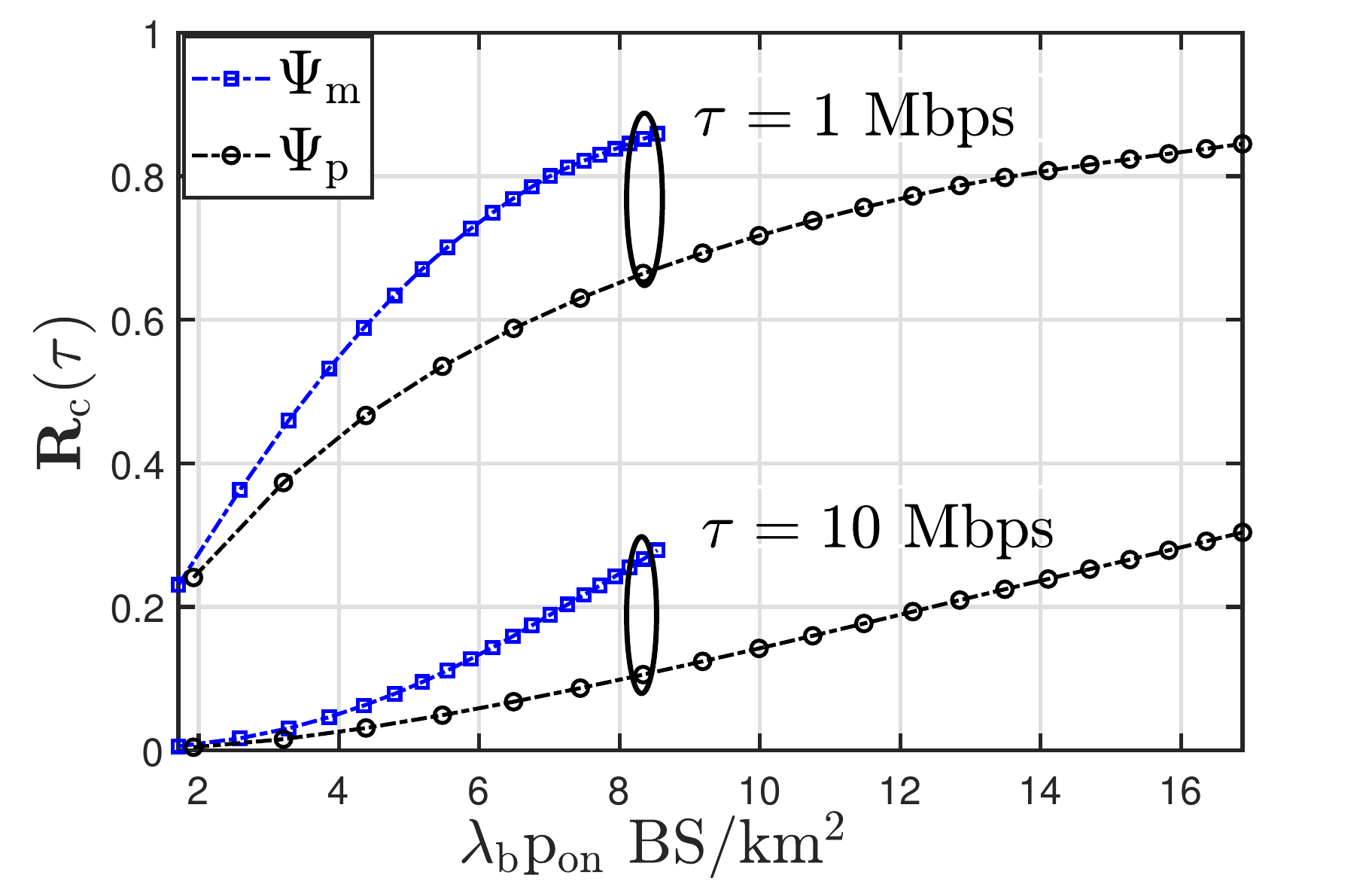} }
	
	\caption{ Impact of the BS density and active BS density on the rate coverage and the active probability for two different values of threshold $\tau$ in PTS and {N-PTS}. Here, $\alpha=3.5$, 
$B=20$ \si{MHz}, and the per-road vehicular density $\lambda=m\lambdaPT=15$ \si{vehicles/km}. Hence, the mean typical load (\ie B/(mean-load on typical cell)) varies between $.5$ \si{MHz/users} to $8$ \si{MHz/users}. 
}
	\label{fig:8}
	\adjustfigspace \adjustfigspace
\end{figure}

\subsection{Rate coverage analysis}
From the results thus far, we have observed that while the PTS has a higher off probability, it also has a higher per-BS load.  Hence, if the system bandwidth is enough to support the load, the PTS is expected to perform better than N-PTS in terms of rate coverage. Fig. \ref{fig:8}(a) shows the variation of rate coverage with respect to BS density for PTS and {N-PTS} with equal vehicular density ($\densityPTS=\densityITS$).
As expected, the rate coverage increases with the BS density. However, at any BS density, the rate coverage in PTS and N-PTS is almost equal, which may appear counter-intuitive. However, it must be noted that the active probability is significantly lower in PTS, which means that PTS can achieve almost the same rate coverage but at a much lower active BS density. This effect can be observed in Fig. \ref{fig:8}(b) which shows the variation of the rate coverage with respect to the active BS density for PTS and N-PTS. Here, we can observe that PTS can achieve significantly higher coverage than N-PTS. Further, the densification of BSs in PTS does not impact the power budget, but still results in an improvement in the rate coverage.  

\section{Conclusion and Future Scope}
In this paper, we have developed a comprehensive approach to the modeling and analysis of platooned vehicular traffic. The approach relies on a novel point process that captures vehicular platooning by explicitly capturing three layers of randomness: (i) irregular layout of the roads by modeling them as a PLP, (ii) randomness in the placement of the platoons on each road by modeling them as a PPP, and (iii) randomness in the location of each vehicle in a platoon by modeling them collectively as an MCP. After deriving several foundational results for this {\em triply stochastic} process, which we called PLP-MCP, we focused explicitly on the V2I communication network for platooned traffic consisting of BSs that serve the platooned traffic. For this setting, we present several key results related to the load distributions on the typical and the tagged BSs. Using these results, we derived the per-user rate coverage of this network and provided a detailed comparative analysis of the PTS and N-PTS cases to understand the effect of vehicular platooning on the performance of the resulting vehicular network. While the rate coverage of these two cases appear similar at the first glance, we defined and studied specific distributional properties of the underlying setup to expose subtle performance trends. Our results collectively demonstrated that the rate coverage of PTS is actually higher when we account for the active BS density. Since this paper presents the very first comprehensive analytical approach to the study of platooned vehicular traffic, there are naturally many extensions possible. Most importantly, it will be interesting to consider an additional tier of roadside units, which are an important component of the emerging vehicular networks. This will require an almost new analysis starting from the distributional results for the load on the typical and the tagged cells. It will also be useful to extend this framework to consider realistic scheduling algorithms.


\appendix

\subsection{Proof of LF of $\Psi_{\mrm}$: Unconditioned and under Palm}\label{Laplcefunctional}
		The LF of $\Psi_{\mrm}$ is
\begin{small}
	\begin{align*}
		&\laplace{\Psi_{\mrm}}(v)=\mathbb{E}_{\Psi_{\mrm}}\left[e^{\left(-\sum_{{\z}_{k,j,i}\in\Psi_{\mrm}} v(\z_{k,j,i})\right)}\right]\\
		&{\stackrel{(a)}=\mathbb{E}_{{\roadPP}}\left[\prod\nolimits_{{\road_{i}\in{\roadPP}}}\mathbb{E}_{\Psi_{\road_{i}},\road_{i}}\left[\prod\nolimits_{\z_{k,j}\in{\Psi_{\road_{i}},\road_{i}}}e^{-v(\z_{k,j})}\right]\right]\stackrel{(b)}=\mathbb{E}_{\roadPP}\left[\prod\nolimits_{\road_{i}\in\roadPP}{G_{\Psi_{\roadf (\rho_{\road_{i}},\phi_{\road_{i}})},\roadf(\rho_{\road_{i}},\phi_{\road_{i}})}(e^{-v})}\right]},
	\end{align*}
\end{small}where $(a)$ is obtained conditioned on ${\roadPP}$, and $(b)$ is obtained by applying the PGFL of $1$D MCP located on a line {$\road_{i}$}. Finally, applying the PGFL of the PLP, we get the LF of $\Psi_{\mrm}$.
For LF under Palm, we assume that the typical point of $\Psi_{\mrm}$ is located at the origin without loss of generality. The LF under Palm  consists of the product of two terms - the LF of an independent and ``unconditioned copy'' of $\Psi_{\mrm}$, and the LF of the MCP on the tagged line
$\road_{\ob}=\roadf\left(0,\phi\right)$. Hence,
\begin{small}
	\begin{align*}
		&\laplace{\Psi_{\mrm}}^{!{\ob}}(v)=\mathbb{E}^{!\ob}\left[\exp\left(- \sum\nolimits_{\z_{k,j,i}\in\Psi_{\mrm}\setminus\{\ob\}} \!\!\!\! v(\z_{k,j,i})\right)\right]=\mathbb{E}_{\Psi_{\mrm}}\left[\exp\left(- \sum\nolimits_{\z_{k,j,i}\in\Psi_{\mrm}}\!\!\!\!v(\z_{k,j,i})\right)\right]\mathbb{E}_{\Psi_{\road},\road}^{\ob!}\left[\exp\left(-v(\z_{k,j,0})\right)\right]\\
		&\stackrel{(a)}=\laplace{\Psi_{\mrm}}(v)\mathbb{E}^{\ob!}_{\Psi_{\road},\road}\left[\prod\nolimits_{\z_{k,j,0}\in \Psi_{\road}\setminus\{\ob\}}\exp\left(-v(\z_{k,j,0})\right)\right]\\
		&\stackrel{(b)}=\laplace{\Psi_{\mrm}}(v)\mathbb{E}_{\Phi}\left[\mathbb{E}_{\Psi_{\road},\road|\phi}^{\ob!}\left[\prod\nolimits_{\bm{z}_{k,j,0}\in \Psi_{\road}\setminus\{\ob\}}\exp\left(-v(\bm{z}_{k,j,0})\right)\right]\right]\stackrel{(c)}=\laplace{\Psi_{\mrm}}(v)\pi^{-1}\int_{0}^{\pi}G_{\Psi_{\road},\road}^{!{\ob}}(\exp(-v))\dv \phi.
	\end{align*}
\end{small}
Here, $(a)$ is obtained by applying the PGFL for $\Psi_{\mrm}$, $(b)$ is due to conditioning on the orientation of tagged line $\roadf\left(0,\phi\right)$, and $(c)$ is obtained by applying the PGFL of $1$D MCP followed by deconditioning over the RV $\phi$, and $G_{\Psi_{\road},\road}^{!{\ob}}(\cdot)$ is given in \eqref{under_palm}.

\subsection{Distribution of $S(r)$}\label{prooftheorem1} 

\textbf{Proof of Theorem \ref{theorem1}:}
The number of vehicles $S(r)$ inside ball $\bt_{2}(\ob,r)$ is 
\[	S(r)=\sum\nolimits_{\road_i\in \roadPP,\rho_{\road_i}\in[-r,r]}N_{\road_i}.\]
Recall that $N_{\road_{i}}=\Psi_{\road_{i}}\left(\bt_{2}(\ob,r)\right)$ denotes the number of vehicles on $\road_i=\roadf\left(\rho_{\road_{i}},\phi_{\road_{i}}\right)$ falling inside  $\bt_{2}(\ob,r)$. 
The condition indicates that the distance of the line from the origin ($\rho_{\road_i}$) needs to be inside the range $[-r,\,r]$ for that line to have at least one point inside $\bt_{2}(\ob,r)$ \cite{chiu2013stochastic}.
 Now, RVs $\{N_{\road_{1}},N_{\road_{2}},\ldots\}$ are independent and identically distributed (iid),
hence PGF of $S(r)$ is
\begin{small}
\begin{align*}
{	\mathcal{P}_{S(r)}(s)}=\mathbb{E}\left[\prod\nolimits_{\road_i\in \roadPP,\rho_{\road_i}\in[-r,r]}	\mathcal{P}_{N_{\road_{i}}}(s,r)\right]=
	\expect{\prod\nolimits_{\road_i\in \roadPP,\rho_{\road_i}\in[-r,r]}	\exp\left(g\left(s,\sqrt{r^2-\rho_{\road_{i}}^2}\right)\right)}
	,
\end{align*}
\end{small}%
where the PGF of $N_{\road_{i}}$ is given by \eqref{PGFMCP}. Since $\rho_{\road_{i}},\phi_{\road_{i}}$ are points of a PPP in $\mathbf{C^{*}}$, 
using PGFL of PPP \cite{SGBook2022}, we get the desired result.

To get probability $\mathbb{P}[S(r)=k]$, we require the $k$th derivative of PGF. 
If we define \begin{small}
	\begin{align*}
		&f_\mrm(s,r)=2\pi{\densityRoads} \int_{0}^{r}\left(\exp (g(s,\sqrt{r^2-\rho^2}))-1\right)\dv \rho,
	\end{align*}%
\end{small}%
PGF {$\mathcal{P}_{S(r)}(s)$} takes the form of $\exp(f_\mrm(s,r))$. Hence, we use the Faà di Bruno's formula \cite{faadibruno} to get \eqref{PSk}. To get $k$th derivative {$f^{(k)}_\mrm(r)$} of {$f_\mrm(s,r)$ at $s=0$}, we need to apply the Faà di Bruno's formula one more time to get \eqref{fsk}.

\textbf{Proof of Corollary \ref{cor1.1}}: 
The first derivative of  {${\mathcal{P}_{S(r)}(s)}$} is  
\begin{small}
	\begin{align}
		&{{\mathcal{P}^{(1)}_{S(r)}(s)}=2\pi{\densityRoads}\mathcal{P}_{S(r)}(s)\frac{\dv}{\dv s}\int_{0}^{r}\left(\expS{g(s,\sqrt{r^2-\rho^2})}\right)\dv \rho.} \label{firstderivative}
	\end{align}
\end{small}
Replacing $s=1$ in \eqref{firstderivative} and solving further, we get the mean of $S(r)$. Similarly, the second derivative $\mathcal{P}^{(2)}_{S(r)}(s)$ of ${\mathcal{P}_{S(r)}(s)}$ is
\begin{small}
	\begin{align*}
		&={\mathcal{P}_{S(r)}(s)}\left(\frac{\dv}{\dv s}\int_{0}^{r}2\pi{\densityRoads}\left(\expS{g(s,\sqrt{r^2-\rho^2})}\right)\dv \rho\right)^{2}+{\mathcal{P}_{S(r)}(s)}\frac{\dv^{2}}{\dv s^{2}}\int_{0}^{r}2\pi{\densityRoads}\left(\expS{g(s,\sqrt{r^2-\rho^2})}\right)\dv \rho\\
		&
		=\left(\mathbb{E}\left[S(r)\right]\right)^{2}+2\pi{\densityRoads}\int_{0}^{r}\left[\left(g^{(1)}(1,\sqrt{r^2-\rho^2})\right)^{2}+g^{(2)}(1,\sqrt{r^2-\rho^2})\right]\dv \rho.
	\end{align*}
\end{small}
Using the second derivative of PGF of $S(r)$, we derive the variance $\mathrm{Var}(S(r))$ of $S(r)$ as
\begin{small}
	\begin{align}\label{integration}
		\mathrm{Var}(S(r))&={\mathcal{P}^{(2)}_{S(r)}(s)}+\mathbb{E}\left[S(r)\right]-\left(\mathbb{E}\left[S(r)\right]\right)^{2}\nonumber\\
		&=2\pi{\densityRoads}\int_{0}^{r}\left[\left(g^{(1)}(1,\sqrt{r^2-\rho^2})\right)^{2}+g^{(2)}(1,\sqrt{r^2-\rho^2})\right]\dv \rho+\densityPTS\pi r^{2}.
	\end{align}
\end{small}
From \eqref{g(0,rho)}, $g^{(1)}(1,\sqrt{r^2-\rho^2})$ and $g^{(2)}(1,\sqrt{r^2-\rho^2})$ is
\begin{small}
	\begin{align}
		&g^{(1)}(1,\sqrt{r^2-\rho^2})=2\lambdaPT\left[\lambda_{\drm}\beta(\sqrt{r^{2}-\rho^2})\times |\sqrt{r^{2}-\rho^2}-a|+\frac{\lambda_{\drm}}{2}{\left(\beta(\sqrt{r^2-\rho^2})\right)^2
		}\right],\label{first_derivative}\\
		&g^{(2)}(1,\sqrt{r^2-\rho^2})=2\lambdaPT\left[\left(\lambda_{\drm}\beta(\sqrt{r^{2}-\rho^2})\right)^{2}|\sqrt{r^{2}-\rho^2}-a|+\frac{\lambda_{\drm}^{2}}{3}{\left(\beta(\sqrt{r^2-\rho^2})\right)^{3}}\right]\label{second_derivative}.
	\end{align}
\end{small}
We can simplify the integrals presented in \eqref{integration} based on the value of $a$ as follows.
\textbf{Case I}: {If $a>r$,}  $\beta(\sqrt{r^2-\rho^2})=2\sqrt{r^2-\rho^2}$. Hence, 
\begin{small}
	\begin{align*}
		&\int_{0}^{r}\left(g^{(1)}(1,\sqrt{r^{2}-\rho^{2}})\right)^2\dv \rho=\frac{32}{3}(a\lambdaPT\lambda_{\drm})^{2}r^3,\,\int_{0}^{r}\left(g^{(2)}(1,\sqrt{r^{2}-\rho^{2}})\right)\dv \rho=8\lambdaPT\lambda_{\drm}^2\left(\frac{2}{3}a r^3-\frac{1}{16}\pi r^4\right).
	\end{align*}%
\end{small}%
Substituting the above values in  \eqref{integration}, we get the variance of $S(r)$ for $a>r$.

\textbf{Case II}:
{If $a<r$,} for  $0<\rho<\sqrt{r^2-a^2}$, $\beta(\sqrt{r^2-\rho^2})=2a$, and when $\sqrt{r^2-a^2}<\rho<r$, $\beta(\sqrt{r^2-\rho^2})=2\sqrt{r^2-\rho^2}$. Hence, 
\begin{small}
	\begin{align*}
		&\int_{0}^{r}\left(g^{(1)}(1,\sqrt{r^{2}-\rho^{2}})\right)^2\dv \rho=\frac{8a^2r^3}{3},\\
		&\int_{0}^{r}\left(g^{(2)}(1,\sqrt{r^{2}-\rho^{2}})\right)\dv \rho=\int_{0}^{\sqrt{r^2-a^2}}\left(g^{(2)}(1,\sqrt{r^{2}-\rho^{2}})\right)\dv \rho+\int_{\sqrt{r^2-a^2}}^{r}\left(g^{(2)}(1,\sqrt{r^{2}-\rho^{2}})\right)\dv\rho.	
	\end{align*}
\end{small}
On substituting $\beta(\sqrt{r^{2}-a^2})$'s value and further manipulations, we get the desired result. 

\subsection{Distribution of $\hat{S}_\mrm(r)$ }\label{prooftheorem2}
\textbf{Proof of Theorem \ref{theorem2}:}
From Theorem \ref{thm:slivE}, we get
\begin{align}
\hat{S}(r)=\vehPP^{!}(\bt_{2}(\ob,r))|\ob\in \vehPP \stackrel{(d)}{=} \vehPP(\bt_{2}(\ob,r)) + \Psi_{\ell_\ob} (\bt_{2}(\ob,r)) + \dauP_{\bm{x}_\ob}(\bt_{2}(\ob,r)),
\end{align}
where the three RVs in RHS are independent. Hence, the PGF of $\hat{S}_\mrm(r)$ is the product of 3 PGFs: the PGF of $S(r)$ $f_{1}(s)={\mathcal{P}_{S(r)}(s)}$, the PGF of $\Psi_{\ell_\ob} (\bt_{2}(\ob,r))$ which is $f_2(s)=\exp\left(g(s,r)\right)$ and the PGF of  $\dauP_{\bm{x}_\ob}(\bt_{2}(\ob,r))$ which is $f_3(s)={a^{-1}}\int_{0}^{a}e^{(s-1)\lambda_{\drm}\A_{1}(r,a,x)}{\dv x}$ \ie 
\begin{small}
	\begin{align}\label{proofofcor2}
		&{\mathcal{P}_{\hat{S}(r)}(s)}=f_1(s,r)f_2(s,r)f_{3}(s,r).
	\end{align}
\end{small}
\textbf{Proof of Corollary \ref{cor_vehicle}}:
Applying generalized Leibniz rule \cite{generlizedlebnizrule} to compute the $k$th derivative of \eqref{proofofcor2} and then from \eqref{eq:pgfpmfrelation}, we get \eqref{PMFhatS}.
The derivative of $f_{2}^{(k)}(s,r)$  can be computed using  Faà di Bruno's formula. Further, 
\begin{small}
	\begin{align*}
		&f_3(s,r)=\int_{0}^{a}a^{-1}e^{(s-1)\lambda_{\drm}\A_{1}(r,a,x)}{\dv x},\,\text{and} \, f_{3}^{(k)}(s,r)=\int_{0}^{a}\left(\lambda_{\drm}
		\A_{1}(r,a,x)\right)^{k}a^{-1}e^{(s-1)\lambda_{\drm}\A_{1}(r,a,x)}{\dv x}.
	\end{align*}%
\end{small}%

\subsection{Proof of Theorem \ref{taggedchord}: Distribution of the tagged chord length}\label{prooftaggechord}
The joint CCDF of the lengths $L_1=\ob\mathbf{Q}_1$ and $L_{2}=\ob\mathbf{Q}_2$ can be written as
\begin{small}
	\begin{align*}
		\overline{F}_{L_{1},L_{2}}(l_1,l_2)=\mathbb{P}(L_{1}>l_{1},L_2>l_2)=	\mathbb{P}[A],
	\end{align*}
\end{small}
where event $
		A=\1\left(\mathbf{Q}_{1},\mathbf{Q}_{2} \text{ and  $\ob$ belong to the same cell}\right)$.
%
If we let $A_i$ be the event that $\mathbf{Q}_{1}$, $\mathbf{Q}_{2}$, and the origin $\ob$, all three locations lie in a single cell $V_{\y_{i}}$ of point $\mathbf{P}_{i}$ located at $\y_i$, 
then
\begin{small}
 \begin{align*}
	\indside{A}=\sum\nolimits_{\y_i\in\Phi_\bs} \indside{A_{i}}.
\end{align*}
	\begin{align*}
	\text{\normalsize Now, }	&A_{i}
		=\1 \left(\vphantom{\frac{text}{den}}\text{$\bt_{2}(\ob,y_{i})$, $\bt_{2}(\mathbf{Q}_{1},|\mathbf{P}_{i}\mathbf{Q}_{1}|)$ and}\vphantom{\frac{text}{den}}\text{ $\bt_{2}(\mathbf{Q}_{2},|\mathbf{P}_{i}\mathbf{Q}_{2}|)$ have no other point except $\mathbf{P}_{i}$}\right)\\
		&=\1\left(\text{$\bt_{2}(\mathbf{Q}_{1},r(l_{1}))$ and $\bt_{2}(\mathbf{Q}_{2},r(l_{2}))$ have no point except $\mathbf{P}_{i}$}\right).
	\end{align*}
\end{small}%
Hence,
\begin{small}
	\begin{align*}
		\overline{F}_{L_{1},L_{2}}(l_1,l_2)&=\prob{A}=\expect{\sum\nolimits_{\y_i\in\Phi_\bs} \indside{A_{i}}}
		\stackrel{(a)}=\lambda_{\bs}\int_{\R^{2}}\mathbb{P}\left(\Phi_\bs \text{ has no point in }\bt_{2}(\mathbf{Q}_{1},r(l_{1}))\cup\bt_{2}(\mathbf{Q}_{2},r(l_{2}))\right)\dv \y\\
		&\stackrel{(b)}={\lambda_{\bs}}\int_{0}^{2\pi}\int_{0}^{\infty}\exp\left(-{\lambda_{\bs}}\mathcal{V}(l_{1}+l_{2},r(l_{1}),r(l_{2}))\right)y \dv y \dv \theta,
	\end{align*}%
\end{small}%
where  $(a)$  is due to the Campbell-Mecke theorem \cite{SGBook2022} and $(b)$ is due to conversion in polar coordinates.
%
Now, we can compute the joint PDF $f_{L_{1},L_{2}}(l_1,l_2)$ from the joint CDF as
\begin{small}
	\begin{align*}
		&f^{''}_{L_1,L_2}(l_1,l_2)=\frac{\partial^2 \overline{F}_{L_{1},L_{2}}(l_1,l_2)}{\partial l_1 \partial l_2},
	\end{align*}
\end{small}%
which gives \eqref{eq:taggedchordjointpdf}. Now, Since $C_\ob=L_1+L_2$, the PDF of $C_\ob$ can be derived by integrating the joint PDF over the line. 

\subsection{Proof of Theorem \ref{theorem44}: {Distribution of load on the typical cell}}\label{prooftheorem44}
The number of points falling in the typical cell 
is
\begin{small}
	\begin{align*}
		S_{\mrm}=\sum\nolimits_{\road_i\in \Phi}\Psi_{\road_{i}}\left(V_{\typical}\right).
	\end{align*}
\end{small}
Let $N$ be the number of chords intersecting with the typical cell $V_{\typical}$. Here, $N$ is a Poisson RV with mean ${\densityRoads}Z$, where $Z$ is also a RV denoting the perimeter of the typical cell. Since the number of points on each chord is iid, the PGF of $S_{\mrm}$ conditioned on $N$ is
\begin{small}
	\begin{align*}
		\mathcal{P}_{S_{\mrm}\vert N=n}(s)&=\mathbb{E}_{\Psi_{\mrm}}\left[s^{S_{\mrm}}\vert N=n\right]
		=\left[\int_{0}^{\infty}\expS{g\left(s,\frac{c}{2}\right)}f_{C}(c)\dv c\right]^{n}.
	\end{align*}
\end{small}%
%
Deconditioning over $n$, we get 
\begin{small}
\begin{align}\label{condtionalPGF}
	&\mathcal{P}_{S_{\mrm}\vert Z=z}(s)=
	\exp\left(-{\densityRoads}z\left(1-\int_{0}^{\infty}\expS{g\left(s,\frac{c}{2}\right)}f_{C}(c)\dv c\right)\right).
\end{align}
\end{small}
Now, deconditioning over the distribution of $Z$ (using the PDF of $Z$ given in \eqref{perimeter}), we obtain the PGF.
From \eqref{eq:pgfpmfrelation}, we can compute
 $\mathbb{P}[S_{\mrm}=k]$ from
 the $k$th derivative of PGF $\mathcal{P}_{S_{\mrm}^{}}(s)$. 
 The $k$th derivative of $\mathcal{P}_{S_{\mrm}^{}|Z}(s)$ is
\begin{small}
	\begin{align}\label{kth_derivative}
		\mathcal{P}^{(k)}_{S_{\mrm}^{}|Z}(s)&=\frac{\dv^{k}}{\dv s^{k}}\left(\exp
		(g_{\mrm}(s))
		\right),
	\end{align}
\end{small}%
where $g_{\mrm}(s)$ is given by
\begin{small}
	\begin{align*}
		&g_{\mrm}(s)=\left(-{\densityRoads}z\left(1-\int_{0}^{\infty}\expS{g\left(s,\frac{c}{2}\right)}f_{C}(c)\dv c\right)\right).
	\end{align*}
\end{small}
As it is in the form of $f(h(s))$, we use the Faà di Bruno's formula \cite{faadibruno} to get
\begin{small}
	\begin{align}
		&\mathcal{P}^{(k)}_{S_{\mrm}^{}|Z}(s)=\exp\left(g_{\mrm}(s)\right)\bellf{}\left( g_{\mrm}^{(1)}(s),\ldots,g_{\mrm}^{(k)}(s)\right),\label{PGF1s}
	\end{align}
\end{small}%
To find the $k$th derivative of $g_{\mrm}(s)$, we apply the Faà di Bruno's formula again and substitute $s=0$ to get the value of  $g_{\mrm}^{(k)}(s)$. Now, deconditioning over $Z$ gives the desired result.

\subsection{Proof of Theorem \ref{Theorem6}: Approximate tagged load distribution}\label{proofTheorem6}
Note that
$\widetilde{M}_\mrm=\Psi'_{\mrm}(\bt_{2}(\ob,R_{\ob}))+\Psi'_{\road_\ob}(\mathsf{C}_\ob)+\dauP'_{\bm{x}_\ob}(\mathsf{C}_\ob)$,
where $\mathsf{C}_\ob=\road_\ob \cap \Vor_\ob=$ is the tagged chord and $\bm{x}_{\ob}$ is the  parent point associated with the typical point. Further note that if the tagged chord has length $c_\ob$,  its center {$\bm{x}_{c_{\ob}}$}    is distributed uniformly  in $\left[-c_{\ob}/2,\,c_{\ob}/2\right]$. We also note that  $\bm{x}_{\ob}$ is uniformly distributed in $\left[-a,\,a\right]$ \cite{pandeykth}. 
We already know the distribution of the first term.  From \eqref{PGFS}, the PGF for $\Psi_{\mrm}\left(\bt_{2}(\ob,R_{\ob})\right)$  is
\begin{small}
	\begin{align}\label{third}
		&{\mathcal{P}_{\Psi_{\mrm}\left(\bt_{2}(\ob,R_{\ob})\right)}(s)}=\mathcal{P}_{S(R_{\ob})}(s).
	\end{align}
\end{small}%

For the second term, note that $\mathsf{C}_\ob=\road_\ob \cap \Vor_\ob = \bt_{1}({\bm{x}_{c_{\ob}}},{c_{\ob}}/{2})$.  Due to stationarity of $\Psi'_{\road_\ob}$ relative to the line $\road_\ob$, $\Psi'_{\road_\ob}(\bt_{1}({\bm{x}_{c_{\ob}}},{c_{\ob}}/{2}))=\Psi'_{\road_\ob}(\bt_{1}(\ob,{c_{\ob}}/{2}))$.
The PGF of RV $\Psi_{\l_{\ob}}\left(\bt_{1}\left(\ob,c_{\ob}/2\right)\right)$ is
\begin{small}
	\begin{align}\label{second}
	{	\mathcal{P}_{\Psi_{\l_{\ob}}\left(\bt_{1}\left(\ob,.5c_{\ob}\right)\right)}(s)}=\mathcal{P}_{N_{\road_{\ob}}}\left(s,c_{\ob}/2\right),
	\end{align}
\end{small}
 where $\mathcal{P}_{N_{\road_{\ob}}}(s,{c_{\ob}}/{2})$ is provided in \eqref{PGFMCP}.

For the third term, we note that $\dauP'_{\bm{x}_\ob}$ can have points only in $\bt_{1}(\bm{x}_\ob,{a}/{2})$. Hence,  $\dauP'_{\bm{x}_\ob}(\mathsf{C}_\ob)$ (\ie the number of points on the tagged chord due to the tagged platoon) varies depending on  $\bm{x}_{\ob}$ and  {$\bm{x}_{c_{\ob}}$}. It can be shown that conditioned on {$\bm{x}_{c_{\ob}}$} and $\bm{x}_{\ob}$, $\dauP'_{\bm{x}_\ob}$ is a PPP in the region $\bt_{1}(\bm{x}_{\ob},a)\cap \bt_{1}({\bm{x}_{c_{\ob}}},{c_{\ob}}/{2})$ with density $\lambda_\mathrm{d}$. 
The mean number of points in this PPP is 
$\lambda_{\drm}\A_{1}\left({c_{\ob}}/{2},a,\left|\bm{x}_{c_\ob}-\bm{x}_{\ob}\right|\right)$. Hence, its PGF is 
%
$\expS{(s-1)\lambda_{\drm}\A_{1}({c_{\ob}}/{2},a,|\bm{x}_{c_{\ob}}-\bm{x}_{\ob}|)}$.
Using the law of total probability, deconditioning over $\bm{x}_{c_{\ob}}$, and $\bm{x}_\ob$, the PGF of the third term is given as
\begin{small}
	\begin{align}\label{afterdecondition}
		&\mathcal{P}_{\dauP'_{\bm{x}_\ob}(\mathsf{C}_\ob)}
		(s|C_\ob=c_{\ob})=\int_{x_{\ob}=-a}^{a}\int_{x_c=-{c_{\ob}}/2}^{{c_{\ob}}/{2}}\frac{1}{2ac_{\ob}}{e^{(s-1)\lambda_{\drm}\A_{1}({c_{\ob}}/{2},a,|x_{c}-x_{\ob}|)}}{}\dv x_c \dv x_{\ob}.
	\end{align}
\end{small}%
Conditioned on $R_{\ob}$ and $C_{\ob}$, the three terms are independent. Therefore, the PGF of $\widetilde{M}_{\mrm}$ is the product of the PGFs of these terms, namely \eqref{third}, \eqref{second} and \eqref{afterdecondition}. Deconditioning over $R_{\ob}$ and $C_{\ob}$, we get the PGF of $\widetilde{M}_{\mrm}$.
\bibliographystyle{IEEEtran}
\bibliographystyle{ieeetran}

\vspace{12pt}

\end{document}